\documentclass[a4paper,10pt]{article}
\usepackage{amssymb,amsbsy,amsmath,amsfonts,amscd,amsthm}
\usepackage[T1]{fontenc}
\usepackage{textcomp}
\usepackage{graphicx}
\usepackage{color}
\usepackage{pdfsync}
\usepackage{enumitem}
\usepackage{hyperref}
\usepackage[useregional]{datetime2}
\usepackage{pgf,tikz}
\usetikzlibrary{arrows}
\usetikzlibrary[patterns]

\newtheorem{theorem}{Theorem}[section]
\newtheorem{lemma}[theorem]{Lemma}
\newtheorem{proposition}[theorem]{Proposition}
\newtheorem{corollary}[theorem]{Corollary}
\newtheorem{definition}[theorem]{Definition}
\theoremstyle{remark}
\newtheorem{remark}[theorem]{Remark}
\newtheorem{remarks}[theorem]{Remarks}

\def\R{{\mathbb R}}
\def\N{{\mathbb N}}
\def\M{{\mathbb M}}

\def\diag{\mathop{\textnormal{diag}}}

\def\cof{\mathop{\textnormal{cof}}}
\def\tr{\mathop{\textnormal{tr}}}
\def\div{\mathop{\textnormal{div}}\nolimits}

\def\SO{\mathrm{SO}}
\def\SL{\mathop{\mathrm{SL}}}
\def\Sl{\mathrm{sl}}
\def\Sym{\mathop{\mathrm{Sym}}}
\def\Skew{\mathop{\mathrm{Skew}}}

\def\bigpenta{\begin{tikzpicture}[scale=0.15,line cap=round,line join=miter,line width=.15mm]
\draw  (2.,0.)-- (3.,0.);
\draw (3.,0.)-- (3.3090169943749475,0.9510565162951532);
\draw  (3.3090169943749475,0.9510565162951532)-- (2.5,1.5388417685876266);
\draw  (2.5,1.5388417685876266)-- (1.6909830056250525,0.9510565162951536);
\draw  (1.6909830056250525,0.9510565162951536)-- (2.,0.);
\end{tikzpicture}}
\def\pentagon{\begin{tikzpicture}[scale=0.11,line cap=round,line join=miter,line width=.08mm]
\draw  (2.,0.)-- (3.,0.);
\draw (3.,0.)-- (3.3090169943749475,0.9510565162951532);
\draw  (3.3090169943749475,0.9510565162951532)-- (2.5,1.5388417685876266);
\draw  (2.5,1.5388417685876266)-- (1.6909830056250525,0.9510565162951536);
\draw  (1.6909830056250525,0.9510565162951536)-- (2.,0.);
\end{tikzpicture}}
\def\penta{\hbox to 0pt{\hss\pentagon\hss}}

\title{A Lagrangian formulation for the Oldroyd B fluid and the second principle of thermodynamics}

\author{Herv\'e Le Dret$^1$ and Annie Raoult$^2$}{}
\date{\today, \DTMcurrenttime}
\begin{document}
\maketitle

\medskip

{\footnotesize
 \centerline{$^1$Sorbonne Universit\'e, Universit\'e Paris Cit\'e, CNRS, Laboratoire Jacques-Louis Lions, F-75005 Paris, France}
} 

\medskip

{\footnotesize
 \centerline{$^2$Universit\'e Paris Cit\'e, CNRS, MAP5, F-75006 Paris, France}
}

\bigskip

\begin{abstract}
We show that the Oldroyd B fluid model is the Eulerian form of a Lagrangian model with an internal variable that satisfies the second law of thermodynamics under some conditions on the initial value of the internal variable. We similarly derive several new nonlinear versions of the Oldroyd B model as well as Lagrangian formulations of the Zaremba-Jaumann and Oldroyd A fluid models. We discuss whether or not these other models satisfy the second law.
\end{abstract}

\section{Introduction}

The purpose of this article is to introduce a physically motivated thermodynamic framework for the Oldroyd B complex fluid model. In this framework, we show that the second law of thermodynamics is satisfied by this model, under some conditions. 

Historically, Oldroyd introduced cases (A) and (B) as  examples of his general theory, generalizing a fluid model of Fröhlich and Sack \cite{Frohlich}, without any concern for thermodynamics, see \cite{Oldroyd}. His main concern was the invariance of the equations of state with respect to coordinate systems. 
In order to ascertain the compatibility of what is now known as the Oldroyd B model with thermodynamics,  this model thus needs to be complemented with an appropriate Helmholtz free energy making it possible to define an internal dissipation, without changing the dynamics itself. Of course, there are infinitely many different possibilities for doing this. 

One of the goals of the present article is precisely to provide such a thermodynamic extension of the Oldroyd B model, in a way that  phenomenologically reflects the Newtonian nature of the solvent fluid and the nonlinearly elastic nature of the polymer particles suspended in the solvent. Building on an idea of Francfort and Lopez-Pamies  \cite{GFOLP}, we introduce a Lagrangian formulation with an internal variable that produces the Oldroyd B dynamics in the Eulerian description. The second goal of the present article is to show that the associated internal dissipation remains always nonnegative for any future evolution, if and only if the initial value of the internal variable belongs to an entirely characterized set. It should be stressed that working in the Lagrangian formulation makes it much easier to achieve this, compared to the equivalent Eulerian formulation.

Compatibility with the second law of thermodynamics is oftentimes mentioned in the Oldroyd B, or more generally complex fluid, literature. Let us mention, among others, the works \cite{Astarita}, \cite{Leonov92}, and for works using the GENERIC framework, \cite{Brunk}, \cite{Dressler}, \cite{Suzuki}. It is however not always easy to discern what is exactly meant by such phrases as ``consistent with thermodynamics'', nor what the thermodynamic framework actually is and what the underlying constitutive assumptions are. We strive to be as explicit as we can in this respect below. 

We place ourselves within the tradition of the Coleman-Noll approach to thermomechanics, \cite{Coleman-Mizel}-\cite{ColemanNoll}. 
Given a set of constitutive assumptions for the Helm\-holtz free energy, the first Piolà-Kirchhoff stress tensor (in the Lagrangian description) and the flow rule for the internal variable, the well-known Coleman-Noll procedure aims at deriving relations and conditions on these constitutive laws that are necessary and sufficient for the Clausius-Duhem inequality to be satisfied by any evolution of the corresponding viscoelastic model, in particular for all initial values of the internal variable. One of the main outcomes of the Coleman-Noll procedure is a constitutive law for the internal dissipation, which is  given in terms of the assumed constitutive laws of the model.

It turns out that this classical approach cannot succeed in the case of the Oldroyd B fluid, see section 5.5 of \cite{HLDAR1}. The appropriate approach is what we call the \emph{conditional Coleman-Noll procedure},  which we introduce here for a general incompressible viscoelastic material with an internal variable. The same approach works, and is technically simpler, for a compressible viscoelastic material with an internal variable, but we need incompressibility for the Oldroyd B fluid application. The main point is that this conditional approach to the Coleman-Noll procedure takes into account the initial value of the internal variable in a crucial, nontrivial way.
Indeed, such a model may satisfy the second law for any evolution of the deformation if and only if the initial value of the internal variable is constrained to belong to some proper subset of the set of all possible initial values.

This article is as follows. We first recall in Section \ref{notation} the notation of \cite{HLDAR1} and review briefly the notion of objective derivative and the classical Oldroyd B model in Sections \ref{les derivees objectives} and \ref{section OB}. We then revisit in Section \ref{cadre general}
 the thermo-visco-elastic theory with internal variables described in \cite{HLDAR1}, this time  in the incompressible case. We introduce a conditional version both of the second law applied to such a model, and of the Coleman-Noll procedure. The condition in question corresponds to the set of initial values for the internal variables for which any future evolution satisfies the second law of thermodynamics. We accordingly develop a conditional version of the Coleman-Noll procedure, here in the incompressible case, in order to derive a constitutive law for the internal dissipation that is adequate for such models. 

Next in Section \ref{section lagrangienne}, we show how the starting idea of  \cite{GFOLP} can be modified in order to obtain the classical Oldroyd~B fluid as a specific instance of such a visco-elastic material. The internal variable we consider is symmetric-matrix valued, which is reminiscent of the conformation or configuration tensors used by several authors in complex fluid modeling, see for instance \cite{Giesekus}, \cite{Hulsen}, \cite{WapHulsen} among others, but is nonetheless slightly different.  

At this point, in Section \ref{le cas ou ca marche}, we apply the previous thermodynamic developments to the new Lagrangian formulation of the Oldroyd B model. In particular, we identify the set of initial conditions that make it satisfy  the second law, namely the set of symmetric positive semi-definite matrices, in both Lagrangian and Eulerian descriptions. This is an if and only if statement: if the initial condition of the internal variable is not positive semi-definite, then there are evolutions for which the dissipation is strictly negative at some point in time. It must be emphasized that in general, the internal variable does not always remain positive semi-definite during evolution, as opposed to what happens in \cite{Hulsen} or \cite{Brunk}. Therefore, this is not just a question of restricting the set of values that the internal variable may take during its evolutions, but a situation that is quite a bit subtler than that.

We conclude the article in Section \ref{pour quelques modeles de plus} with various  generalizations of the Oldroyd B model and their connection with the second law, that are obtained by the same Lagrangian approach. These include
 nonlinear versions of the Oldroyd B model, among which is the quadratic one obtained in \cite{GFOLP}, and the Zaremba-Jaumann and Oldroyd A fluid models.

\section{General notation}\label{notation}
As in \cite{HLDAR1}, we use the convention of denoting any quantity pertaining to the Lagrangian description with an uppercase letter and the corresponding Eulerian quantity with the corresponding lowercase letter. We also differentiate between a given quantity and a  constitutive law for that same quantity by using a hat or a tilde for the latter, \emph{e.g.} $A_m$ for the Helmholtz free energy specific density as opposed to $\widehat A_m$ for a constitutive law for this energy. On the Lagrangian side of things, starting from a deformation mapping $(X,t)\mapsto \phi(X,t)$, we will use as thermodynamic variables the deformation gradient $F(X,t)=\nabla_X\phi(X,t)$ and the deformation rate  $H(X,t)=\nabla_XV(X,t)=\frac{\partial F}{\partial t}(X,t)$, where $V$ is the velocity of particules. We will also need the acceleration gradient $\Gamma(X,t)=\frac{\partial H}{\partial t}(X,t)$. In the present article, we will ignore thermal effects. More precisely, we assume that mechanical effects are decoupled from thermal effects, so that we can only consider the mechanical side of thermodynamics. A more complete coupling can easily be added, see \cite{HLDAR1}. 

We will always work in an incompressible setting, which for simplicity is expressed by the fact that $\det F(X,t)=1$ in the Lagrangian description. 

On the Eulerian side, and following the above convention, we let $v(x,t)=V(X,t)$ for the Eulerian velocity and $h(x,t)=\nabla_xv(x,t)=H(X,t)F^{-1}(X,t)$ for its gradient, with the understanding that $(x,t)=\Phi(X,t)=(\phi(X,t),t)$. We also let $d(x,t)=\frac12\bigl(h(x,t)+h(x,t)^T\bigr)$ for the stretching tensor and $w(x,t)=\frac12\bigl(h(x,t)-h(x,t)^T\bigr)$ for the spin tensor. Incompressibility in the Eulerian description is expressed by $\div_x v(x,t)=0$.

We denote the first Piolà-Kirchoff stress tensor by $T_{\mathrm{R}}$ and the Cauchy stress tensor by $\sigma$, which are related by $\sigma(x,t)=T_{\mathrm{R}}(X,t)F^T(X,t)$ at all corresponding space-time points, using the fact that $\det F(X,t)=1$. Even though the Cauchy stress is an Eulerian quantity, it will sometimes be expedient to look at it from the Lagrangian point of view. In particular,  we will need its material derivative $\dot\sigma=\frac{\partial\sigma}{\partial t}+v_i\frac{\partial\sigma}{\partial x_i}$, which is just $\dot\sigma=\frac{\partial}{\partial t}(\sigma\circ\Phi)$. In terms of thermodynamic state functions, we use the Helmholtz free energy with specific density $A_m$, and set the reference volumic mass to $1$ for simplicity. 

We denote the set of $3\times 3$ matrices by $\M_3$, which is endowed with the Frobenius inner product $F:G=\tr(F^TG)$. We let $\M_3^+$ be the subset of matrices with strictly positive determinant, $\Sym_3$ the set of symmetric matrices, $\Sym_3^+$ (resp. $\overline{\Sym_3^+}$) the set of symmetric positive definite (resp. positive semi-definite) matrices, $\Skew_3$ the set of skew-symmetric matrices, $\SL(3)$ the set of matrices with determinant $1$, $\Sl(3)$ the set of trace-free matrices and $\SO(3)$ the set of rotation matrices. For any $M\in\M_3$, $\cof M$ denotes the cofactor matrix of $M$, and $\Sym(M)$ and $\Skew(M)$ respectively denote the symmetric and skew-symmetric parts of $M$. For $F\in\SL(3)$, there holds $F^{-1}=\cof F^T$.

We let $T_F\SL(3)=\{H\in \M_3; \cof F:H=0\}$ denote the tangent space of $\SL(3)$ at $F\in\SL(3)$, with of course $T_I\SL(3)=\Sl(3)$. Accordingly,  $T\SL(3)$ denotes the tangent bundle of $\SL(3)$ that consists of pairs $(F,H)$ with $F\in\SL(3)$ and $H\in T_F\SL(3)$. The tangent space of $T\SL(3)$ at $(F,H)$ is identified with 
$T_F\SL(3)\times T_F\SL(3)$. Observe that, by incompressibility, we have $(F(X,t),H(X,t))\in T_F\SL(3)$ for any deformation $\phi$.

For any differentiable function $\widehat A\colon T\SL(3)\to \R$, we identify the differential $D_{(F,H)}\widehat A$ with a pair of matrices $\frac{\partial \widehat A}{\partial F}(F,H)$ and $\frac{\partial \widehat A}{\partial H}(F,H)$ using the Frobenius inner product: $\langle D_{(F,H)}\widehat A,(K,L)\rangle=\frac{\partial \widehat A}{\partial F}(F,H):K+\frac{\partial \widehat A}{\partial H}(F,H):L$, with $\bigl(\frac{\partial \widehat A}{\partial F}(F,H),\frac{\partial \widehat A}{\partial H}(F,H)\bigr)\in T_F\SL(3)\times T_F\SL(3)$.

\section{Objective derivatives}\label{les derivees objectives}

Objective derivatives occur in situations when one wishes to differentiate Eulerian vector or tensor quantities, typically the Cauchy stress tensor, with respect to time, in a way that is compatible with frame-indifference. 
In Lagrangian terms, we are considering two deformations $\phi$ and $\phi^*$ that are related via
\begin{equation}\label{hyp cinematique AIM}
\phi^*(X,t)=Q(t)\phi(X,t)+a(t),
\end{equation}
where $Q$ and $a$ are arbitrary regular functions with values respectively in  $\SO(3)$ and $\R^3$. We express \eqref{hyp cinematique AIM} in Eulerian terms as $(x^*,t)=(Q(t)x+a(t),t)$.

The principle of frame-indifference, for which we refer to \cite{Truesdell-Noll}, requires that the corresponding Cauchy stresses must satisfy
\begin{equation}\label{expression AIM}
\sigma^*(Q(t)x+a(t),t)=Q(t)\sigma(x,t)Q(t)^T,
\end{equation}
or $\sigma^*(x^*,t)=Q(t)\sigma(x,t)Q(t)^T$ with obvious notation.

Loosely speaking, an objective derivative is a differential operator that is of first order in time and depends on $h$ in such a way that it transforms as above in the same circumstances. It can thus be used for constitutive purposes to derive frame-indifferent models of differential type.

\begin{definition}A first order in time differential operator $\bigpenta$ is an objective derivative, or is objective, if 
\begin{equation*}
\overset{\pentagon^*}{\sigma^{\smash{*}}}(x^*,t)=Q(t)\overset{\penta}{\sigma}(x,t)Q(t)^T,
\end{equation*}
for all functions $\sigma$, $Q$ and $a$ with values in $\Sym_3$, $\SO(3)$ and $\R^3$ respectively, where $\sigma^*$ and $\sigma$ are related via \eqref{expression AIM}.
\end{definition}

A constitutive differential equation for the Cauchy stress can then be assumed of the form
$$
\overset{\penta}{\sigma}(x,t)=G(\sigma(x,t),d(x,t))
$$
for instance in the simplest cases, and if $G$ is itself frame-indifferent, \emph{i.e.}, $G(Q\sigma Q^T,QdQ^T)=QG(\sigma,d)Q^T$, then the resulting differential constitutive law will be compatible with the principle of frame-indifference.

Many classical objective derivatives are of the form
\def\Ob{\mathrm{Ob}}
\begin{equation}\label{forme a priori}
\overset{\penta}{\sigma}=\dot \sigma+\Ob(\sigma,h),
\end{equation}
with $\Ob\colon\Sym_3\times \,\M_3\to\Sym_3$. 
Now, a function $\Ob_s\colon \Sym_3\times\Sym_3\to\Sym_3$ is said to be objective if 
\begin{equation}\label{fonction objective}
\Ob_s(Q\sigma Q^T, QdQ^T)=Q\Ob_s(\sigma , d)Q^T,
\end{equation}
for all $(\sigma,d)\in\Sym_3\times\Sym_3$ and $Q\in\SO(3)$. We note that the set of such objective functions  is entirely characterized, see \cite{Smith}. In particular, it contains all symmetric-valued polynomials in $(\sigma,d)$. The following result is stated in \cite{Gurtinandco}.

\begin{proposition}\label{description des derivees objectives}
An operator of the form \eqref{forme a priori} is objective if and only if 
\begin{equation}\label{forme a posteriori}
\Ob(\sigma ,h)=\sigma w- w\sigma +\Ob_s(\sigma,d).
\end{equation}
with $w=\Skew(h)$, $d=\Sym(h)$, and $\Ob_s\colon \Sym_3\times\Sym_3\to\Sym_3$ is an objective function. 
\end{proposition}

Common examples where $\Ob_s$ is a simple symmetric-valued polynomial in $(\sigma,d)$ are:
\begin{itemize}
\item The Zaremba-Jaumann derivative, introduced in \cite{Zaremba},
$\overset{\square}{\sigma}=\dot{\sigma}+\sigma w-w\sigma,$ with $\Ob_s=0$, which is thus the simplest objective derivative in this sense. 
\item The Oldroyd A or lower convected derivative 
$\overset{\vartriangle}{\sigma}=\dot \sigma+h^T\!\sigma+\sigma h$ with $\Ob_s(\sigma,d)= d\sigma  +\sigma d$.
\item The Oldroyd B or upper convected derivative 
$\overset{\triangledown}{\sigma}=\dot \sigma-h\sigma-\sigma h^T$ with $\Ob_s(\sigma,d)= - d\sigma  -\sigma d$, cases A and B both introduced in \cite{Oldroyd}.
\end{itemize}
Note that the notation is not universal, and neither is the vocabulary, with corotational, covariant and contravariant rates also being in use for these objective derivatives.
We will encounter more objective derivatives of the same general form involving objective polynomials $\Ob_s$ in Section~\ref{nonlinear oldB}.

It is well known that the Oldroyd B derivative appears naturally, without reference to objectivity, by time differentiation of the Cauchy stress expressed with the second Piolà-Kirchhoff stress in the incompressible case, with the simple formula below.

\begin{proposition}\label{et voila OB}Let $\Sigma=F^{-1}\sigma \cof F$ be the second Piolà-Kirchhoff stress. For all deformations, we have
\begin{equation*}\label{apparition d'OB}
\overset{\triangledown}{\sigma}=F\frac{\partial\Sigma}{\partial t}F^T.
\end{equation*}
\end{proposition}

This result may induce a slight preference for the Oldroyd B derivative  among all arbitrary objective derivatives. Moreover, the second Piolà-Kirchhoff stress is especially well suited to the Lagrangian formulation and ensuing thermodynamical study of the Oldroyd B fluid model as we will see below.
\section{The Oldroyd B complex fluid model}\label{section OB}

We give a very brief description of the Oldroyd B fluid model, introduced in~\cite{Oldroyd}. We refer to \cite{Hinch} for historical and physical insights and \cite{Renardy} for a review of mathematical results pertaining to this model. The Oldroyd B model is a model for an incompressible viscoelastic fluid that is supposed to be a dilute suspension of polymer molecules in a Newtonian fluid solvent. It is a model of differential type, see \cite{Truesdell-Noll}, in the sense that the Cauchy stress is not expressed as a function of thermodynamic variables by means of a constitutive law, but is given by a first order differential equation in time that involves the Oldroyd B derivative
\begin{equation}\label{le modele}
\sigma+\lambda_1\overset{\triangledown}{\sigma}=2\eta\bigl(d+\lambda_2 \overset{\triangledown}{d}\bigr),
\end{equation}
where $\eta>0$ is a global viscosity coefficient and $\lambda_1,\lambda_2>0$ are relaxation times. For the model to be physically relevant, it is assumed that $\lambda_2\le \lambda_1$. It is frame-indifferent by construction. Of course, the tensor $\sigma$ above does not include the indeterminate pressure term $-pI$  that is the Lagrange multiplier of the incompressibility constraint $\tr(d)=0$. In the sequel, by Cauchy stress, we will mean Cauchy stress modulo the indeterminate pressure as long as no initial value is specified for the differential equation \eqref{le modele}.

There is a classical additive decomposition of the Cauchy stress that simplifies equation \eqref{le modele}, namely $\sigma=\sigma_s+\sigma_p$, obtained by letting $\eta_s=\frac{\lambda_2}{\lambda_1}\eta$ and $\eta_p=\bigl(1-\frac{\lambda_2}{\lambda_1}\bigr)\eta$ and
\begin{equation}\label{les deux stress}
\sigma_s=2\eta_sd\text{ and }\sigma_p+\lambda_1\overset{\triangledown}{\sigma}_p=2\eta_pd,
\end{equation}
where $\sigma_s$ is the Newtonian solvent stress with solvent viscosity $\eta_s$ and $\sigma_p$ is interpreted as a polymer stress with polymer viscosity $\eta_p$. Conversely, \eqref{les deux stress} implies \eqref{le modele} with  $\eta=\eta_s+\eta_p$ and  $\lambda_2=\frac{\lambda_1\eta_s}{\eta_s+\eta_p}$ and the two formulations are equivalent. 

There are many different ways of deriving the Oldroyd B model from various hypotheses. We concentrate below on a phenomenological  Lagrangian approach, which makes it possible to test the compatibility of the Oldroyd B model with the second law of thermodynamics.

\section{The conditional Coleman-Noll procedure for incompressible viscoelastic materials with internal variables}\label{cadre general}
We continue in the spirit of the general thermo-visco-elastic framework developed in \cite{HLDAR1}, by using the thermodynamic variables $F$ and $H$, complemented by a $\R^k$-valued\footnote{To fix ideas. We just need equation \eqref{KVM flow rule} to make sense.} internal variable $\Xi$ for some $k$, without thermal effects. This framework relies heavily on the exploitation of the second law of thermodynamics in the form of the Clausius-Duhem inequality via the Coleman-Noll procedure, here in an incompressible context. We primarily work here in the Lagrangian description for convenience, and will navigate more equally between Lagrangian and Eulerian descriptions in the specific case of the Oldroyd B fluid later on. 

The definition of such an incompressible viscoelastic model rests on a set of a priori constitutive assumptions. For simplicity, we assume that the material is homogeneous so that there is no dependence on the material point $X$. As already mentioned, the reference volumic mass is set equal to $1$. We are thus given a constitutive law for the Helmholtz free energy specific density of our material $\widehat A_m\colon T\SL(3)\times\R^k\to\R$, so that the free energy density\footnote{For completeness, we should add a term to the free energy that is only a function of temperature, so that entropy balance can also be written down. In the cases of interest herein, all the mechanics is decoupled from thermal effects, which explains why we can afford not to write it here.} at point 
$(X,t)$ is given by 
\begin{equation}\label{lc de Helmoltz}
A_m(X,t)=\widehat A_m(F(X,t),H(X,t),\Xi(X,t)).
\end{equation}
We also are given a constitutive law $\widehat T_{\mathrm{R}}\colon T\SL(3)\times\R^k\to\M_3$ for part of the first Piolà-Kirchhoff stress,  so that likewise 
\begin{equation}\label{lc de PK1}
T_{\mathrm{R}}(X,t)=\widehat T_{\mathrm{R}}(F(X,t),H(X,t),\Xi(X,t))-p(X,t)\cof F(X,t),
\end{equation}
where $p(X,t)$ is a pressure term that is indeterminate in the sense that is not given by a constitutive law using the thermodynamic variables, due to the incompressibility constraint. It can only be determined once the dynamics initial-boundary value problem with given body forces is eventually solved. It will thus remain indeterminate here for our constitutive purposes.

 Finally, we are given a  differential constraint for the internal variable of the form 
\begin{equation}\label{KVM flow rule}
\frac{\partial \Xi}{\partial t}(X,t)=\widehat K(F(X,t),H(X,t),\Xi(X,t)),
\end{equation}
where the flow rule $\widehat K\colon T\SL(3)\times\R^k\to\R^k$ is the last constitutive ingredient of the model. Since the internal variable is assumed to satisfy an ordinary differential equation in time for all $X$, it also needs an initial value $\Xi_0\in\R^k$ at some arbitrary initial time $t_0$ for each $X$ in order to be determined for a given deformation $\phi$. In this sense, $\Xi$ is not strictly speaking an independent thermodynamic variable per se, but its initial value plays such a role.

As mentioned in the Introduction, the classical Coleman-Noll procedure cannot succeed in the case of the Oldroyd B fluid, see section 5.5 of \cite{HLDAR1}. We introduce below a \emph{conditional Coleman-Noll procedure} for a general incompressible viscoelastic material with an internal variable, that takes into account the initial value of the internal variable in a crucial, nontrivial way. This conditional procedure will succeed in the Oldroyd B case, see Section \ref{le cas ou ca marche}.

To be more precise, given some reference configuration $\Omega\subset\R^3$, we say that a deformation $\phi\colon\bar\Omega\times\R\to\R^3$ is admissible if it is of class $C^1$ and piecewise $C^2$ in time, $\det\nabla_X\phi(X,t)=1$ for all $X\in\Omega$ and $t\ge t_0$, and for all $t\ge t_0$, $\phi(\cdot,t)$ is a diffeomorphism between $\bar\Omega$ and $\overline{\phi(\Omega,t)}$. Here we have fixed an arbitrary initial time $t_0$.

The following quantity is known as the internal dissipation,
\begin{equation}\label{def dint}
D_{\text{int}}(X,t)= - \frac{\partial A_m}{\partial t}(X,t)+T_{\mathrm{R}}(X,t):\nabla_X V(X,t).
\end{equation} 
The internal dissipation is a power term which appears as a source term in the heat equation when thermal effects are also taken into consideration, or equivalently, when divided by the temperature, as an internal entropy source. 

The Clausius-Duhem inequality reduces here to the mechanical part of the Clausius-Planck inequalities, which simply reads
 \begin{equation}\label{encore une CD}
D_{\text{int}}(X,t)\ge0\text{ for all $(X,t)$.}
  \end{equation} 
The second law of thermodynamics stipulates that this inequality must hold for the evolutions that can actually occur. Which evolutions are allowed by the second law for a given model then becomes  a constitutive matter. Indeed, the internal dissipation can be computed for all future deformations and all initial values of the internal variable by replacing the quantities in \eqref{def dint} by their expressions in terms of the constitutive laws. This leads us to the following definition.

%
%
\begin{definition}\label{ze big def}We will say that a viscoelastic model \eqref{lc de Helmoltz}, \eqref{lc de PK1}
and \eqref{KVM flow rule} satisfies the second law for an initial value of the internal variable $\Xi(.,0)$ if, given any admissible deformation, the corresponding internal dissipation is nonnegative for all $t\ge t_0$. Letting $\mathcal{C}=\{\Xi_0\in\R^k; \Xi_0=\Xi(X,0)\text{ such that the above holds}\}$,\footnote{Using the constitutive laws \eqref{lc de Helmoltz}, \eqref{lc de PK1}
and \eqref{KVM flow rule} and the definition of the internal dissipation \eqref{def dint}, it is easy to see that this set does not depend on $X\in \Omega$. If the material is not supposed to be homogeneous, the corresponding sets $\mathcal{C}_X$ may depend on $X$.} we say that the second law is \emph{$\mathcal{C}$-conditionally} satisfied by the model.
\end{definition}

Note that the condition in question is a condition that bears on a thermodynamic variable, and not on the constitutive laws. 
According to the definition, if a given model satisfies the second law conditionally, such a model should only be used with initial values $\Xi_0\in\mathcal{C}$ at each point $X$. However, it is of utmost importance to notice that this does not preclude $\Xi(\cdot,t)$ from exiting the set $\mathcal{C}$ at some time $t_1\ge t_0$, still with nonnegative internal dissipation at all subsequent times. In other words, \emph{we do not assume the set $\mathcal{C}$ to be invariant under the flow} 
\eqref{KVM flow rule}.
In particular, we are not in a special case of \cite{HLDAR1} where such an invariance was assumed, apart from the considerations of incompressibility which were also not treated therein. Incompressibility only adds technical complications but no conceptual difference. We will see later on that the Oldroyd B fluid satisfies the conditional second law with a non invariant condition set.

From now on, we thus assume that we have an incompressible viscoelastic model that satisfies the second law conditionally as defined in Definition \ref{ze big def}.

The idea of the Coleman-Noll procedure, which we revisit below, is to plug arbitrary deformations and internal variables into the constitutive laws, compute the corresponding internal dissipation via \eqref{def dint} using \eqref{lc de Helmoltz}, \eqref{lc de PK1}, \eqref{KVM flow rule}, and extract all the constitutive consequences of \eqref{encore une CD}. The arbitrary deformation step is quite often not performed very carefully. It is better to proceed with caution as in \cite{HLDAR1} by explicitly constructing a sufficient set of such deformations, even more so in the incompressible case here. These deformations are solutions of the dynamics equation simply by adjusting the body forces.  

Let us start with such a construction of an extension of any admissible deformation that has an arbitrary acceleration gradient right after some given  time. We will say that a time $t_1\ge  t_0$ is \emph{regular} for a pair $(\phi, X_0)$ if  $t\mapsto \phi(X_0,t)$ is of class $C^2$ in a neighborhood of $t_1$. By assumption, the set of times that are not regular is discrete.

\begin{proposition}\label{construction prudente}Let $X_0\in\Omega$, $\phi$ an admissible deformation and $t_1\ge t_0$ a regular time. Then, for all $\Gamma^*\in T_{F(X_0,t_1)}\SL(3)$, there exists an admissible deformation $\phi^*$ that extends $\phi$ for $t\ge t_1$ and such that $\Gamma^*(X_0,t)\to \Gamma^*\text{ when }t\to t_1^+$.
\end{proposition}

\begin{proof}Let us consider such an admissible deformation $\phi$ and time $t_1\ge t_0$. We wish to modify $\phi$ after time $t_1$ appropriately. To this effect, take any $M\in\Sl(3)$ and set
$$\phi^*(X,t)=e^{((t-t_1)_+)^2M}\phi(X,t)\text{ for all }X\in\Omega.$$
This clearly is an incompressible deformation that agrees with $\phi$ for $t\le t_1$. Moreover,
$$F^*(X,t)=e^{((t-t_1)_+)^2M}F(X,t)\to F(X,t_1)\text{ when }t\to t_1.$$
 Likewise for the deformation rate,
\begin{multline*}
H^*(X,t)=2(t-t_1)_+Me^{((t-t_1)_+)^2M}F(X,t)+e^{((t-t_1)_+)^2M}H(X,t)\\\to H(X,t_1)\text{ when }t\to t_1.
\end{multline*}
The deformation $\phi^*$ is clearly piecewise $C^2$ in time and is thus an admissible deformation that extends $\phi$ for $t>t_1$.  Computing its acceleration gradient at $X_0$ for $t>t_1$ small enough, we get
\begin{multline*}
\Gamma^*(X_0,t)=2\bigl(1+2(t-t_1)^2M\bigr)Me^{(t-t_1)^2M}F(X_0,t)\\
+4(t-t_1)Me^{(t-t_1)^2M}H(X_0,t)+e^{(t-t_1))^2M}\Gamma(X_0,t)\\\to 2MF(X_0,t_1)+\Gamma(X_0,t_1)\text{ when }t\to t_1^+.
\end{multline*}
By construction, $2MF(X_0,t_1)+\Gamma(X_0,t_1)\in T_{F(X_0,t_1)}\SL(3)$. Given any $\Gamma^*\in T_{F(X_0,t_1)}\SL(3)$, we can set $M=\frac12(\Gamma^*-\Gamma(X_0,t_1))\cof F^T(X_0,t_1)\in\Sl(3)$ to ensure that 
$\Gamma^*(X_0,t)\to \Gamma^*\text{ when }t\to t_1^+$.
\end{proof}

The first step in the conditional Coleman-Noll procedure is to write down the master inequality derived from \eqref{encore une CD}.

\begin{proposition}\label{ze master ineq}If a viscoelatic model \eqref{lc de Helmoltz}-\eqref{lc de PK1}-\eqref{KVM flow rule}  $\mathcal{C}$-conditionally satisfies the second law, there holds
\begin{multline}\label{mas ineq}
\Bigl(\widehat T_{\mathrm{R}}(F(X,t),H(X,t),\Xi(X,t))-\frac{\partial\widehat A_m}{\partial F}(F(X,t),H(X,t),\Xi(X,t))\Bigr):H(X,t)\\
-\frac{\partial\widehat A_m}{\partial H}(F(X,t),H(X,t),\Xi(X,t)):\Gamma(X,t)\\
-\frac{\partial\widehat A_m}{\partial \Xi}(F(X,t),H(X,t),\Xi(X,t))\cdot\widehat K(F(X,t),H(X,t),\Xi(X,t))\ge 0
\end{multline}
for all admissible deformations $\phi$, all $X\in\Omega$, all $\Xi(.,0)\in\mathcal{C}$ and at all corresponding regular times $t\ge t_0$.
\end{proposition}

\begin{proof}Let us apply the chain rule to equation \eqref{lc de Helmoltz} at such a regular time  $t$. This yields
\begin{multline}\label{derive de Helm}
\frac{\partial A_m}{\partial t}(X,t)=\frac{\partial\widehat A_m}{\partial F}(F(X,t),H(X,t),\Xi(X,t)):H(X,t)\\
+\frac{\partial\widehat A_m}{\partial H}(F(X,t),H(X,t),\Xi(X,t)):\Gamma(X,t)\\
+\frac{\partial\widehat A_m}{\partial \Xi}(F(X,t),H(X,t),\Xi(X,t))\cdot\frac{\partial \Xi}{\partial t}(X,t),
\end{multline}
where the final inner product $\cdot$ is that of $\R^k$. The result then follows from substituting the flow rule \eqref{KVM flow rule} in the last term in \eqref{derive de Helm}, and from the constitutive law \eqref{lc de PK1} for the second term in the definition of the internal dissipation \eqref{def dint}. Indeed, the indeterminate pressure does not contribute to the dissipation since
$$-p(X,t)\cof F(X,t):H(X,t)=0$$
because $H(X,t)\in T_{F(X,t)}\SL(3)$.
\end{proof}

At this point, it should be fairly clear that the material point $X$ only plays the role of a parameter and that we can most often omit it from the notation in order to save a little space.

As in \cite{HLDAR1}, there is a natural decomposition of the constitutive law for the first Piolà-Kirchhoff stress
\begin{equation}\label{decomposition du stress}\widehat T_{\mathrm{R}}(F,H,\Xi)=\widehat T_{\mathrm{Rd}}(F,H,\Xi)+\frac{\partial\widehat A_m}{\partial F}(F,H,\Xi),
\end{equation}
where $\widehat T_{\mathrm{Rd}}=\widehat T_{\mathrm{R}}-\frac{\partial\widehat A_m}{\partial F}$ is the constitutive law for the kinematically viscous stress. Note that there will in general be no reason not to have $H$ as an argument in $\widehat A_m$.  

We can now clarify the relationship between the condition set $\mathcal{C}$ and the constitutive laws. 
We let 
\begin{multline}\label{def de C0}\mathcal{C}_0=\Bigl\{\Xi_0\in\R^k;\widehat T_{\mathrm{Rd}}(F,H,\Xi_0):H
-\frac{\partial\widehat A_m}{\partial \Xi}(F,H,\Xi_0)\cdot\widehat K(F,H,\Xi_0)\ge 0\\
\text{and }\frac{\partial\widehat A_m}{\partial H}(F,H,\Xi_0)=0,\text{ for all }(F,H)\in T\SL(3)\Bigr\},
\end{multline}
a set which is defined by the constitutive laws.
\begin{proposition}\label{carac de C}
We have $\mathcal{C}\subset\mathcal{C}_0$.
\end{proposition}

\begin{proof}Let us take $F\in \SL(3)$ and $H\in T_F\SL(3)$. The deformation $\phi(X,t)=Fe^{(t-t_0)F^{-1}H}X$ is then clearly admissible and such that $F(X,t_0)=F$ and $H(X,t_0)=H$. Using the same proof as that of Proposition \ref{construction prudente} with an arbitrary $\Gamma^*\in T_F\SL(3)$, we can modify it into another admissible deformation $\phi^*$ that has the same deformation gradient and rate, and $\Gamma^*$ as acceleration gradient at $t=t_0$. It follows that, 
$$-\frac{\partial\widehat A_m}{\partial H}(F,H,\Xi_0):\Gamma^*+
\widehat T_{\mathrm{Rd}}(F,H,\Xi_0):H
-\frac{\partial\widehat A_m}{\partial \Xi}(F,H,\Xi_0)\cdot\widehat K(F,H,\Xi_0)\ge 0.
$$
Taking $\Gamma^*=\lambda \frac{\partial\widehat A_m}{\partial H}(F,H,\Xi_0)$ and letting $\lambda\to+\infty$, we obtain the second part in the definition of $\mathcal{C}_0$. What is left in the above inequality is then the first part in this definition.
\end{proof}


Let us pursue the conditional Coleman-Noll procedure with the derivation of a form of constitutive law for the internal dissipation. Let us define $\widehat D_{\textnormal{int}}\colon T\SL(3)\times\R^k\to \R$ by
\begin{equation}\label{def de lc de diss}
\widehat D_{\textnormal{int}}(F,H,\Xi)=\widehat T_{\mathrm{Rd}}(F,H,\Xi):H
-\frac{\partial\widehat A_m}{\partial \Xi}(F,H,\Xi)\cdot\widehat K(F,H,\Xi).
\end{equation}

\begin{proposition}\label{dissipation conditionnelle}Let us be given an admissible deformation $\phi$ and an initial value $\Xi_0\in\mathcal{C}$. Then at all times $t\ge t_0$, we have
\begin{equation}\label{lc de diss}
D_{\textnormal{int}}(X,t)=\widehat D_{\textnormal{int}}(F(X,t),H(X,t),\Xi(X,t)).
\end{equation}
\end{proposition}

Note that this is valid only conditionally, but this is sufficient for all further developments. 
\begin{proof}The proof is very similar to that of Proposition \ref{carac de C}. Let us fix a regular time $t_1\ge t_0$, and take an arbitrary $\Gamma^*\in T_{F(t_1)}\SL(3)$ (omitting $X$ for brevity). Let $\phi^*$ be given by Proposition \ref{construction prudente}. Using the same initial value $\Xi_0$ for the internal variable, we obtain an evolution $\Xi^*(t)$ of the internal variable that agrees with $\Xi(t)$ up to time $t_1$. Moreover, by the conditional second law hypothesis, the internal dissipation of the corresponding evolution $(\phi^*,\Xi^*)$ is nonnegative at all times.  

Consequently, inequality \eqref{mas ineq} implies that  
\begin{multline*}
\widehat T_{\mathrm{Rd}}(F^*(t),H^*(t),\Xi^*(t)):H^*(t)
-\frac{\partial\widehat A_m}{\partial H}(F^*(t),H^*(t),\Xi^*(t)):\Gamma^*(t)\\
-\frac{\partial\widehat A_m}{\partial \Xi}(F^*(t),H^*(t),\Xi^*(t))\cdot\widehat K(F^*(t),H^*(t),\Xi^*(t))\ge 0
\end{multline*}
for all $t>t_1$. Letting $t\to t_1^+$, it follows by continuity that
\begin{multline*}
\widehat T_{\mathrm{Rd}}(F(t_1),H(t_1),\Xi(t_1)):H(t_1)
-\frac{\partial\widehat A_m}{\partial H}(F(t_1),H(t_1),\Xi(t_1)):\Gamma^*\\
-\frac{\partial\widehat A_m}{\partial \Xi}(F(t_1),H(t_1),\Xi(t_1))\cdot\widehat K(F(t_1),H(t_1),\Xi(t_1))\ge 0.
\end{multline*}
Since $\Gamma^*$ in arbitrary in $T_{F(t_1)}\SL(3)$ and $\frac{\partial\widehat A_m}{\partial H}(F(t_1),H(t_1),\Xi(t_1))\in T_{F(t_1)}\SL(3)$, it follows that 
$$
\frac{\partial\widehat A_m}{\partial H}(F(t_1),H(t_1),\Xi(t_1))=0,
$$
so that equation \eqref{lc de diss} holds at such instants $t_1$. The proof is concluded by observing that both sides of this equation are continuous with respect to $t$.
\end{proof}

\begin{remark}Note that, as opposed to the usual Coleman-Noll procedure, this last relation does not imply that $\widehat A_m$ must not depend on $H$. Its partial gradient is shown to vanish only along trajectories of the evolution that keep the dissipation nonnegative. 

However, the restriction of $\widehat A_m$ to $T\SL(3)\times\mathcal{C}$ does not depend on $H$, since this is the case in $T\SL(3)\times\mathcal{C}_0$. 
In particular, if in addition $\mathcal{C}$ is invariant by the flow, then we recover the usual Coleman-Noll procedure by just restricting the internal variable to be $\mathcal{C}$-valued. Of course, the interesting case is when $\mathcal{C}$ is not invariant by the flow.

It is however probably more practical, and without loss of generality, to only consider constitutive laws $\widehat A_m$ that do not depend on $H$ as we do in the next section.
\end{remark}

\begin{remark}\label{reecriture de C0}
We have a more compact rewriting of definition \eqref{def de C0} as
\begin{equation}\label{def de C0 bis}\mathcal{C}_0=\Bigl\{\Xi_0\in\R^k;\!\widehat D_{\textnormal{int}}(F,H,\Xi_0)\ge 0,\!\frac{\partial\widehat A_m}{\partial H}(F,H,\Xi_0)=0,\text{for all}\,(F,H)\in T\!\SL(3)\Bigr\}.
\end{equation}
\end{remark}

\begin{remark}We have obtained necessary conditions for a constitutive model to satisfy the conditional second law. However,  there is no obvious way of determining the set $\mathcal{C}$ based on the constitutive laws in general.  This is again to be contrasted with the usual Coleman-Noll procedure which produces necessary and sufficient conditions on the constitutive laws for the satisfaction of the unconditional second law, see \cite{HLDAR1}. We will see that the necessary condition obtained above are also sufficient in the case of the Oldroyd B fluid in Section \ref{le cas ou ca marche}, but not of the Oldroyd A and Zaremba-Jaumann fluids in Section \ref{ZJ and co}. 
\end{remark}

\section{A Lagrangian formulation for the Oldroyd B model}\label{section lagrangienne}

By convention, in the sequel, when we write an equality between an  Eulerian quantity and a Lagrangian quantity, it will be meant at corresponding space-time points $(x,t)=(\phi(X,t),t)$. This  keeps the length of formulas under control. 

The original idea of \cite{GFOLP} was to use the standard generalized materials formalism, see \cite{Halphen}, to derive a Lagrangian model that would correspond to an Oldroyd B model in the Eulerian description,  satisfy the second law by construction and have a variational structure. This was not entirely successful, resulting instead in a quadratic version of the Oldroyd B model, see Section \ref{nonlinear oldB}, but not in the Oldroyd B model itself.

 We retain the starting idea of \cite{GFOLP} of using a nonlinearly  elastic neo-Hookean energy
\begin{equation*}\label{neohook}
\widehat W(F)=\frac\mu2\|F\|^2,\; \mu>0.
\end{equation*}
The incompressible neo-Hookean energy is commonly used to model rubbers, and phenomenologically accounts here for the presence of the polymer particules in the suspension.
 This energy  is  frame-indifferent and can be rewritten as $\widetilde W(C)=\frac\mu2\tr C$, where $C=F^TF$ is the usual strain or Cauchy-Green tensor.
We use this nonlinearly elastic stored energy function to define the Lagrangian constitutive law $\widehat A_m$ of the Helmholtz free energy specific density of our model by
\begin{equation}\label{KVM ansatz}
\widehat A_m(F,\Xi{})=\widetilde W(\Xi\, C)=\frac\mu2\Xi{}:C,
\end{equation}
where $\Xi{}$ is a symmetric matrix-valued, dimensionless internal variable,  \emph{i.e.,} $\R^k\approx \Sym_3$ in the general formalism of Section \ref{cadre general}. Then, 
\begin{equation}\label{derivees KVM ansatz}
\frac{\partial\widehat A_m}{\partial F}(F,\Xi{})=\mu \Bigl(F\,\Xi-\frac{\tr \Xi}{\tr C^{-1}}\cof F\Bigr)\text{ and }\frac{\partial\widehat A_m}{\partial \Xi{}}(F,\Xi{})=\frac\mu2C.
\end{equation}
In the above formula, $\frac{\partial\widehat A_m}{\partial F}(F,\Xi{})$ is the orthogonal projection of the unconstrained partial gradient $\mu F\,\Xi$ of $\widehat A_m$ onto the tangent space $T_{F}\SL(3)$ as required by the remarks at the end of Section \ref{notation}.

This form of free energy is inspired by more general choices of internal variables  and elastic energies $\widehat W$, see \cite{HLDAR1} for details. Note that  $\Xi{}$ is symmetric but $\Xi\,C$ is not. The modeling role of $\Xi$ is to account for the interaction between the global deformation of the fluid and its elastic polymer component. Finally, this free energy does not depend on $H$ by assumption.

Let us stress that our choice of internal variable $\Xi$ is very close to what several authors in complex fluid modeling call conformation or configuration tensors, see \cite{Hulsen} for instance, but that it should not be conceived of as some kind of right Cauchy-Green or Finger tensor of some intermediate deformation, which is not assumed to be relevant here. 

We take here the simplest kinematically viscous stress possible, which is that of the Newtonian fluid with viscosity $\eta_s$ in the Lagrangian description, 
\begin{equation}\label{newtonian}
\widehat T_{\mathrm{Rd}}(F,H,\Xi{})=2\eta_s\Sym(HF^{-1})F^{-T}.
\end{equation}
It does not depend on $\Xi$ and will take care of the solvent stress. 

Using  \eqref{derivees KVM ansatz}, the expression \eqref{lc de PK1} of the first Piolà-Kirchhoff stress in terms of its (partial) constitutive law and the natural decomposition  \eqref{decomposition du stress} of the said constitutive law, translate as a natural decomposition of  the Cauchy stress $\sigma=\sigma_s+\sigma_p- p'I$ into a Newtonian viscous stress, the solvent part
\begin{equation}\label{newtonian euler}
\sigma_s=T_{\mathrm{Rd}}F^T=2\eta_s d,
\end{equation}
 a second part given by a constitutive law
\begin{equation}\label{polymer stress}
 \sigma_p=\mu F\Xi{}F^T,
 \end{equation}
 which will turn out to exactly correspond  to the polymer part
 in decomposition~\eqref{les deux stress},
 and finally a pressure term $-p' I$ with 
 $$p'=p+\mu\frac{\tr\Xi}{\tr C^{-1}},
 $$
where $p'$ is another indeterminate pressure,\footnote{We will ignore the difference between $p'$ and $p$ in the sequel.} which is again not given by a constitutive law.

Let us define $\Sigma_p=F^{-1}\sigma_p F^{-T}$ to serve as the second Piolà-Kirchhoff stress part corresponding to $\sigma_p$. Clearly, 
\begin{equation}\label{Z simple}
\Sigma_p=\mu \Xi{}.
\end{equation}
The numerical value of the modulus $\mu$ will become irrelevant, but we keep it for reasons of dimensional homogeneity. 

To complete the identification of the Oldroyd B fluid as a visco-elastic material with an internal variable, we just need to specify the flow rule. 
We thus choose
\begin{equation}\label{la bonne flow rule}
\widehat K(F,H,\Xi{})=-\frac1{\lambda_1}\Xi{}+\frac{2\eta_p}{\mu\lambda_1}F^{-1}\Sym(HF^{-1})F^{-T}.
\end{equation}
In other words, the ordinary differential equation \eqref{KVM flow rule} for the internal variable here assumes the form
\begin{equation}\label{puisqu'il faut insister...}
\frac{\partial \Xi{}}{\partial t}=-\frac1{\lambda_1}\Xi{}+\frac{2\eta_p}{\mu\lambda_1}F^{-1}\Sym(HF^{-1})F^{-T}.
\end{equation}

It is easy to check that with these choices, the resulting incompressible material is frame-indifferent. It clearly has a symmetric Cauchy stress constitutive law. It can also be directly shown in this Lagrangian description that the corresponding material is fluid, see \cite{HLDAR1}. This is however not really necessary since,

\begin{proposition}\label{jackpot!}
The visco-elastic material defined by \eqref{KVM ansatz}, \eqref{newtonian}, and \eqref{la bonne flow rule} is the Oldroyd B fluid with material constants $\lambda_1$, $\lambda_2$ and $\eta$.
\end{proposition}

\begin{proof}
We have already seen that $\sigma_s=2\eta_s d$ by \eqref{newtonian euler}. Besides, 
$\sigma_p$ is given by equation \eqref{polymer stress}. Proposition \ref{et voila OB} applied to 
\eqref{polymer stress} and \eqref{Z simple}
then shows that
\begin{equation}\label{derivee Old B du stress p}
\overset{\triangledown}{\sigma}_p=\mu F\frac{\partial \Xi{}}{\partial t}F^T.
\end{equation}
We now substitute the ordinary differential equation \eqref{puisqu'il faut insister...} in the above relation, and obtain
\begin{align*}
 \overset{\triangledown}{\sigma}_p&=\mu F\Bigl(-\frac1{\lambda_1}\Xi{}+\frac{2\eta_p}{\mu\lambda_1}F^{-1}\Sym(HF^{-1})F^{-T}\Bigr)F^T\\
 &=-\frac{\mu}{\lambda_1} F\Xi{}F^T+\frac{2\eta_p}{\lambda_1}\Sym(HF^{-1})\\
 &=-\frac{1}{\lambda_1} \sigma_p+\frac{2\eta_p}{\lambda_1}d,
\end{align*}
or in other words
\begin{equation}\label{flowruleeulerien}
\sigma_p+\lambda_1\overset{\triangledown}{\sigma}_p=2\eta_p d,
\end{equation}
which is the polymer part of the Oldroyd B fluid constitutive differential equation. As already seen in Section \ref{section OB}, this is equivalent to $$\sigma+\lambda_1\overset{\triangledown}{\sigma}=2\eta(d+\lambda_2\overset{\triangledown}{d}),$$
with $\sigma=\sigma_s+\sigma_p$, $\eta=\eta_s+\eta_p$ and  $\lambda_2=\frac{\lambda_1\eta_s}{\eta_s+\eta_p}$. This is the Oldroyd B fluid model with global viscosity $\eta$ and relaxation times $\lambda_1$ and $\lambda_2$.
\end{proof}

\begin{remarks}We have shown that the specific instance of visco-elastic materials with internal variables described above satisfies the Oldroyd B equation in the Eulerian description. Conversely, given any Oldroyd B fluid, we can manufacture such a material that reproduces its behavior. Indeed, it suffices to take $\eta_s=\frac{\lambda_2}{\lambda_1}\eta$,  $\eta_p=\bigl(1-\frac{\lambda_2}{\lambda_1}\bigr)\eta$, and any nonzero value for $\mu$ in the Lagrangian model. 

It is a posteriori interesting that the thermodynamically motivated decomposition of the first Piolà-Kirchhoff stress \eqref{decomposition du stress} actually corresponds to the Cauchy stress decomposition for the Oldroyd B fluid \eqref{les deux stress}, which initially looked like little more than an algebraic trick.
\end{remarks}

In \cite{HLDAR1}, we rephrased the standard Oldroyd B model as a thermodynamics model with an internal variable expressed in Eulerian form. We took $\xi=\sigma_p$ for the internal variable, a free energy $\widehat a_m$ only function of $\xi$, and used the differential constitutive law \eqref{flowruleeulerien} itself as a flow rule. A special case of it turns out to be equivalent to the present, more physically grounded, Lagrangian approach. Indeed, by \eqref{polymer stress}
$
\xi=\mu F\Xi{}F^T$ or again $\Xi{}=\frac1\mu F^{-1}\xi F^{-T}$.
As said above in terms of Eulerian flow rule, we simply took
as ordinary differential equation for $\xi$
\begin{equation}\label{flow rule OB 1}
\dot{\xi}=h\xi+\xi h^T+\frac1{\lambda_1}\bigr({-}\xi+2\eta_pd\bigl),
\end{equation}
a mere rewriting of \eqref{flowruleeulerien}. We have already noted it to be equivalent to \eqref{puisqu'il faut insister...}.

\begin{proposition} The Eulerian expression of the free energy \eqref{KVM ansatz} is given by  
$\widehat a_m(\xi)=\frac12\tr\xi$.
\end{proposition}

\begin{proof} We have
$$\widehat A_m(F,\Xi{})=\frac\mu2 \Xi{}:C=\frac\mu2\tr(\Xi{}F^TF)=\frac12\tr(\mu F\Xi{}F^T)=\frac12\tr\xi,$$
hence the Eulerian form of the free energy.
\end{proof}

With this specific choice of Eulerian free energy and internal variable, we are thus recovering our Lagrangian model expressed in Eulerian terms, the flow rule of which is basically just the Oldroyd B differential constitutive law. The advantages of the Lagrangian formulation will appear below in dealing with the second law of thermodynamics.

\section{The Oldroyd B fluid and the second law of thermodynamics}\label{le cas ou ca marche}

As already outlined in the Introduction, the compatibility of the Oldroyd B fluid model with thermodynamics is not a trivial issue. For instance, when considered as a model with no internal variable, then the internal dissipation must be the naive one $d_{\mathrm{int,naive}}=\sigma:d$, here expressed in Eulerian terms, see \cite{HLDAR1}. There we showed numerically that $d_{\mathrm{int,naive}}$ can become strictly negative after some time, even if initially strictly positive.

Let us thus develop the thermodynamics of the Lagrangian formulation of the Oldroyd B model presented in Section \ref{section lagrangienne}.
\begin{proposition}\label{scoop!}
The constitutive law for the Lagrangian internal dissipation is
\begin{equation}\label{dissipation Old B}
  \widehat D_{\mathrm{int}}(F,H,\Xi{})=2\eta_s\|\Sym(HF^{-1})\|^2+\frac{\mu}{2\lambda_1}\Xi{}:C.
\end{equation}
\end{proposition}

\begin{proof}
 Since $\frac{\partial\widehat A_m}{\partial H}=0$ by assumption, we do not need the full force of Proposition \ref{dissipation conditionnelle}. The chain rule directly shows that formulas  \eqref{def de lc de diss} and 
\eqref{lc de diss} hold for all $(F,H)\in T\SL(3)$ and $\Xi\in \Sym_3$.

Consider first the Newtonian fluid term,
\begin{align*}
 \widehat T_{\mathrm{Rd}}(F,H,\Xi{}):H&=2\eta_s\bigl(\Sym(HF^{-1})F^{-T}\bigr):H
 =2\eta_s\tr\bigl(F^{-1}\Sym(HF^{-1})H\bigr)\\
 &=2\eta_s\tr\bigl(\Sym(HF^{-1})HF^{-1}\bigr)=2\eta_s\Sym(HF^{-1}):\bigl(HF^{-1}\bigr)\\
 &=2\eta_s\|\Sym(HF^{-1})\|^2,
\end{align*}
which is to be expected from the corresponding Eulerian expression.
Next, we look at the dissipation coming from the internal variable. It follows from \eqref{derivees KVM ansatz} that 
\begin{align*}
-\frac{\partial\widehat A_m}{\partial \Xi{}}(F,\Xi{}): \widehat K(F,H,\Xi{})&=-\frac\mu2 C:\Bigl(-\frac1{\lambda_1}\Xi{}+\frac{2\eta_p}{\mu\lambda_1}F^{-1}\Sym(HF^{-1})F^{-T}\Bigr)\\
&=\frac\mu{2\lambda_1} C:\Xi{}-\frac{\eta_p}{\lambda_1}L(H),
\end{align*}
where $L(H)=C:{}\bigr(F^{-1}\Sym(HF^{-1})F^{-T}\bigl)$. This linear form vanishes since
\begin{multline*}
L(H)=\tr\bigr(CF^{-1}\Sym(HF^{-1})F^{-T}\bigl)\\
=\tr\bigr(F^{T}\Sym(HF^{-1})F^{-T}\bigl)
=\tr\bigr(\Sym(HF^{-1})\bigl)=0.
\end{multline*}
Indeed, $\Sym(HF^{-1})=d$ and $\tr d=0$ by incompressibility. This completes the proof of relation \eqref{dissipation Old B}.
\end{proof}

On the Eulerian side, see \cite{HLDAR1}, we had $\dot\xi=\widehat k(h,\xi)$ with 
$
\widehat k(h,\xi)=h\xi+\xi h^T+\frac1{\lambda_1}\bigr({-}\xi+2\eta_pd\bigl),
$
\emph{i.e.}, \eqref{flow rule OB 1}, and the internal dissipation
\begin{equation*}\label{un clausius-planck de plus}
\widehat d_{\mathrm{int}}(h,\xi)=\sigma:d-\frac{\partial \widehat a_m}{\partial\xi}(\xi):\widehat k(h,\xi)\text{ with }\sigma=2\eta_sd+\xi.
\end{equation*}
Let us apply this here with $\widehat a_m(\xi)=\frac12\tr\xi$.

\begin{proposition}\label{pareil en eulerien}
 The Eulerian constitutive law for the internal dissipation is given by
 \begin{equation}\label{dissipation eulerienne}
  \widehat d_{\mathrm{int}}(h,\xi)=2\eta_s\|d\|^2+\frac1{2\lambda_1}\tr \xi.
 \end{equation}
\end{proposition}

\begin{proof}
 We have here $\frac{\partial \widehat a_m}{\partial\xi}(\xi)=\frac12I$, so that
 $$
 \widehat d_{\mathrm{int}}(h,\xi)=2\eta_sd:d+\xi:d-\frac12\tr(\widehat k(h,\xi)),
 $$
 with 
 $$
 \tr(\widehat k(h,\xi))=2\xi:h-\frac1{\lambda_1}\tr \xi=2\xi:d-\frac1{\lambda_1}\tr \xi,
 $$
 by  incompressibility and the symmetry of $\xi$. Therefore, \eqref{dissipation eulerienne} holds true.
\end{proof}

\begin{remark}
Note that, even though the internal dissipation is not written as a specific density, nonetheless at all corresponding values of the thermodynamic variables, $
\widehat D_{\mathrm{int}}(F,H,\Xi{}) =\widehat d_{\mathrm{int}}(h,\xi)
 $ because of incompressibility, since both refer to the same model. In particular, we have already noticed that $\tr\xi=\tr\sigma_p=\mu \Xi{}:C$. Moreover, since $\tr\sigma_s=0$ and $\sigma=\sigma_s+\xi$, we can also write $\widehat d_{\mathrm{int}}(h,\xi)=2\eta_s\|d\|^2+\frac1{2\lambda_1}\tr\sigma$.
 \end{remark}

Our goal now is to show that our Lagrangian formulation of the Oldroyd B model satisfies the second law conditionally with a condition set $\mathcal{C}$ that we will entirely identify.
We need the following standard lemma.

\begin{lemma}\label{une inegalite matricielle}
Let $B$ and $C$ be two symmetric positive semi-definite $n\times n$ matrices. Then we have
\begin{equation}\label{ladite inegalite}
B:C\ge n(\det B)^{\frac1n}(\det C)^{\frac1n}.
\end{equation}
In particular, $B:C\ge0$.
\end{lemma}

\begin{proof}
We first remark that
$$
B:C=\tr(BC)=\tr(B^{\frac12}B^{\frac12}C^{\frac12}C^{\frac12})=\tr(C^{\frac12}B^{\frac12}B^{\frac12}C^{\frac12})=\|B^{\frac12}C^{\frac12}\|^2.
$$
 Let $M=B^{\frac12}C^{\frac12}$ and $M_i\in\R^n$ be its column vectors. We have
$$B:C=\sum_{i=1}^n\|M_i\|^2\ge n\Bigl(\prod_{i=1}^n\|M_i\|^2\Bigr)^{\frac1n}\ge n(\det M)^{\frac2n}= n(\det B)^{\frac1n}(\det C)^{\frac1n},$$
by the inequality between the arithmetic and geometric means and by the Hadamard inequality.
\end{proof}
Note that inequality \eqref{ladite inegalite} is sharp. We proceed to identify the set $\mathcal{C}_0$.

\begin{proposition}\label{condition C}
 We have $\mathcal{C}_0=\overline{\Sym_3^+}$.
\end{proposition}

\begin{proof}
It follows from \eqref{def de C0 bis} and Proposition \ref{scoop!} that
\begin{multline*}\label{C0 Old B}\mathcal{C}_0=\Bigl\{\Xi_0\in\R^k;2\eta_s\|\Sym(HF^{-1})\|^2+\frac{\mu}{2\lambda_1}\Xi_0:C\ge 0\\
\text{ for all }(F,H)\in T\SL(3)\Bigr\},
\end{multline*}
Let $\Xi_0\in\mathcal{C}_0$.
Taking $H=0$, it follows that
$$
\Xi_0:C\ge 0,
$$
for all $C\in\Sym_3^+\cap\SL(3)$.
Multiplying by any positive factor and by continuity, we also have $\Xi_0:C\ge 0$ for all $C\in\overline{\Sym_3^+}$. Diagonalizing $\Xi_0$ in the form $\Xi_0=Q\Delta Q^T$ with $Q\in SO(3)$ and $\Delta$ diagonal, we see that $\Delta:C\ge 0$ for all $C\in\overline{\Sym_3^+}$. Choosing $C=\diag(1,0,0)$, we obtain that $\Delta_{11}\ge 0$, and similarly for the other eigenvalues of $\Xi_0$. Hence $\mathcal{C}_0\subset\overline{\Sym_3^+}$.

Conversely, if $\Xi_0\in \overline{\Sym_3^+}$, then Lemma \ref{une inegalite matricielle} shows that $\Xi_0:C\ge 0$ for all $C\in \overline{\Sym_3^+}$, hence a fortiori, all $C\in\Sym_3^+\cap\SL(3)$. Consequently,  $\overline{\Sym_3^+}\subset\mathcal{C}_0$.
\end{proof}

The next step is our main result.
\begin{proposition}\label{cool}
Our Lagrangian version of the Oldroyd B model satisfies the second law of thermodynamics $\overline{\Sym_3^+}$-conditionally.
\end{proposition}

\begin{proof}We denote  partial derivatives with respect to time with a prime, as time is the single relevant variable in what follows, and of course omit $X$.

By Propositions \ref{carac de C} and \ref{condition C}, we know that the condition set is such that $\mathcal{C}\subset\overline{\Sym_3^+}$. Conversely, let us assume that the initial condition of the internal variable $\Xi{}(0)$ is positive semi-definite (set $t_0=0$ for simplicity). We consider an arbitrary admissible deformation. It is well known that 
$$
C'(t)=2F^T(t)d(t)F(t).
$$
Let us rewrite the second term in the flow rule \eqref{la bonne flow rule} using this remark:
$$
F^{-1}\Sym(HF^{-1})F^{-T}=\frac12F^{-1}F^{-T}C'F^{-1}F^{-T}=\frac12C^{-1}C'C^{-1}=-\frac12(C^{-1})'.
$$
The ordinary differential equation \eqref{puisqu'il faut insister...} for $\Xi{}$ is then rewritten as
\begin{equation}\label{insistons}
\Xi'=-\frac1{\lambda_1}\Xi{}-\eta_*(C^{-1})',
\end{equation}
where $\eta_*=\frac{\eta_p}{\mu\lambda_1}\ge0$. This differential equation
is linear with continuous right-hand side, therefore the Cauchy problem is well-posed on $\R_+$ for any initial value $\Xi{}(0)\in\overline{\Sym_3^+}$.\footnote{Of course, the well-posedness holds more generally for any $\Xi(0)\in\M_3$.} Moreover, it has constant coefficients, thus the Duhamel formula  provides an expression for $\Xi{}$,
$$
\Xi{}(t)=e^{-\frac t{\lambda_1}}\Xi{}(0)-{\eta_*} e^{-\frac t{\lambda_1}}\int_0^te^{\frac s{\lambda_1}}
(C^{-1})'(s)\,ds,
$$
 Integrating the second term by parts, we obtain
$$
e^{\frac t{\lambda_1}}\Xi{}(t)=\Xi{}(0)+{\eta_*}\Bigl(C^{-1}(0)-e^{\frac t{\lambda_1}}C^{-1}(t)+\frac1{\lambda_1}\int_0^te^{\frac s{\lambda_1}}
C^{-1}(s)\,ds\Bigr).
$$
Consequently, 
\begin{align*}
e^{\frac t{\lambda_1}}\Xi{}(t):C(t)&=\Xi{}(0):C(t)+{\eta_*} C^{-1}(0):C(t)-3{\eta_*} e^{\frac{t}{\lambda_1}}\\
&\qquad\qquad\qquad\qquad+\frac{\eta_*}{\lambda_1} \int_0^te^{\frac s{\lambda_1}}
C^{-1}(s):C(t)\,ds\\
&=\Xi{}(0):C(t)+{\eta_*} \bigl(C^{-1}(0):C(t)-3\bigr)\\
&\qquad\qquad\qquad\qquad+\frac{\eta_*}{\lambda_1} \int_0^te^{\frac s{\lambda_1}}
\bigl(C^{-1}(s):C(t)-3\bigr)\,ds.
\end{align*}
By Lemma \ref{une inegalite matricielle}, we first have $\Xi{}(0):C(t)\ge 0$ for all $t\ge 0$ since $\Xi{}(0)$ and $C(t)$ are both symmetric positive semi-definite. Secondly, $C^{-1}(s)$ and $C(t)$ are symmetric positive definite and belong to $\SL(3)$, by incompressibility. Therefore, also by Lemma \ref{une inegalite matricielle}, we have $C^{-1}(s):C(t)\ge3$ for all $s\ge 0$ and consequently $\Xi{}(t):C(t)\ge 0$ for all $t\ge 0$. The nonnegativity of the internal dissipation then follows from formula \eqref{dissipation Old B}.
\end{proof}

It is quite remarkable that this case of conditional second law is one in which the condition set $\mathcal{C}$ is not stable by the flow, as the following explicit example shows.

\begin{proposition}\label{damned}
If $\eta_p>0$, there is an admissible deformation such that $\Xi{}(X,0)\in\overline{\Sym_3^+}$, but $\Xi{}(X,t)\notin\overline{\Sym_3^+}$ for some $t>0$. If $\eta_p=0$, then $\Xi{}(X,0)\in\overline{\Sym_3^+}$ implies that $\Xi{}(X,t)\in\overline{\Sym_3^+}$ for all $t\ge0$.
\end{proposition}

\begin{proof}
Assume first that $\eta_p>0$.
We take  $\phi(X,t)=F(t)X$ with $F(t)=\diag(e^t,e^{-t},1)$. This is obviously an admissible deformation, which corresponds to a steady Eulerian flow $h(t)=d(t)=\diag(1,-1,0)$. Assuming $\lambda_1=1$ without loss of generality, equation \eqref{puisqu'il faut insister...} becomes
 in this case, 
 $$
 \bigl(e^t\Xi{}(t)\bigr)'=\frac{2\eta_p}\mu\diag(e^{-t},-e^{3t},0).
 $$
 Integrating this between $0$ and $t$, we obtain
$$
\Xi{}(t)=e^{-t}\Xi{}(0)+\frac{2\eta_p}\mu \diag\Bigl(e^{-t}-e^{-2t},\frac13\bigl(e^{-t}-e^{2t}\bigr),0\Bigr).
$$
If $\eta_p>0$, then for any positive semi-definite $\Xi{}(0)$, the smallest eigenvalue of $\Xi{}(t)$ thus goes to $-\infty$ when $t\to+\infty$. Consequently, $\Xi{}(t)$ exits $\overline{\Sym_3^+}$ in finite time (and stays out of $\overline{\Sym_3^+}$ for all subsequent times).

If $\eta_p=0$, then for all admissible deformations, $\Xi{}(t)=e^{-\frac t{\lambda_1}}\Xi{}(0)$ for all $t$. Therefore, the flow rule obviously preserves positive semi-definiteness in this particular case.
\end{proof}

\begin{remark}
In this example, we obtain that
$$
\Xi{}(t):C(t)=e^{-t}\Xi{}(0):C(t)+\frac{2\eta_p}\mu\Bigl(e^t+\frac13e^{-3t}-\frac43\Bigr)\ge 0
$$
as expected for all $t$.

Note that when $\eta_p=0$, the result of Proposition \ref{cool} is trivial in view of Lemma \ref{une inegalite matricielle} and the second part of Proposition \ref{damned}.
\end{remark} 

See also another example, albeit a numerical one, of evolution that does not preserve positive semi-definiteness in Section \ref{ZJ and co}, Figure \ref{fig2}.

\begin{remark}\label{borne inf}We remark that equation \eqref{insistons} may be rewritten as
$$
 \bigr(\Xi+\eta_*C^{-1}\bigl)'=-\frac1{\lambda_1} \bigr(\Xi+\eta_*C^{-1}\bigl)+\frac{\eta_*}{\lambda_1}C^{-1}.
$$
Since $C^{-1}$ is always positive definite, it follows that if $\Xi(0)\in\mathcal{C}_0$, then $\Xi+\eta_*C^{-1}$ is also always positive definite. In other words, $\Xi(t)>-\eta_*C^{-1}(t)$ in the sense of quadratic forms. This lower bound does not seem to help in proving that the dissipation remains nonnegative though. In Eulerian terms, it follows that 
$$\xi(t)>-\frac{\eta_p}{\lambda_1}I$$
as soon as $\xi(0)$ is positive semi-definite. A similar lower bound was observed in \cite{Brunk}.
\end{remark}

The following is then a direct consequence of Proposition \ref{pareil en eulerien}.

\begin{corollary}\label{pas de changement cote Euler}The Eulerian dissipation $d_{\mathrm{int}}$ remains nonnegative for all times and all given velocity fields if and only if $\sigma_p(0)$ is positive semi-definite.
\end{corollary}

Of course, the naive dissipation $\sigma:d$ will still change sign for some velocity fields. The benefits of working with the Lagrangian formulation are quite clear: We just have to deal with a simple linear ordinary differential equation, with the space variable $X$ only a parameter, for which we can use Duhamel's formula to get an explicit solution. We do not have to worry about the intricacies of the Oldroyd B derivative in the Eulerian picture.

The question now is whether it is legitimate to only consider initial values for the internal variable that ensure the conditional second law. For a general model, the physical meaning of a given internal variable may be unclear. It can be a question of principle to only pick such initial conditions. In the particular case of the Oldroyd B fluid, we actually have
$$
\Xi{}(0)=\frac1\mu\Sigma_p(0)=\frac1\mu F^{-1}(0)\sigma_p(0)F^{-T}(0),
$$
or directly in the Eulerian description
$$\xi(0)=\sigma_p(0).$$
Now of course, the polymer stress is not entirely well defined from the constitutive point of view, unless an initial value is chosen for it. Furthermore, since the model is incompressible, the actual physical Cauchy stress is of the form
$$\sigma(0)=2\eta_sd(0)+\sigma_p(0)-p(0)I,$$
where $p$ is the indeterminate pressure. So the quantity that is physically meaningful is $\sigma_p(0)-p(0)I=\sigma(0)-2\eta_sd(0)$. From the constitutive point of view, we are at liberty to incorporate some of the indeterminate pressure into $\sigma_p(0)$ in such a way that it becomes positive semi-definite without changing the right-hand side, and thus consider the second law to be satisfied in all situations. 

If we were considering an initial-boundary value problem that was well-posed, then the pressure would be determined by the problem data and the above liberty would not necessarily be available. 

\begin{remark}A natural question is whether or not the present formulation of the Oldroyd B model admits a dissipation potential in the sense of standard generalized materials, \cite{GFOLP}-\cite{Halphen}, or in the somewhat different sense introduced in \cite{HLDAR1}.

In the case $\eta_p>0$, the answer is clearly negative, otherwise the second law would hold unconditionally. The case $\eta_p=0$ can be considered as unconditional if we restrict the internal variable to $\overline{\Sym_3^+}$, which is then invariant under the flow rule. This is a convex set and the function $\widehat P_{\mathrm{diss}}\colon
 \M_3^+\times \M_3\times\overline{\Sym_3^+}\times\overline{\Sym_3^+}\to\R$ defined by
$$\widehat P_{\mathrm{diss}}(F,H,\Xi{},\Lambda)=\eta_s\|\Sym(HF^{-1})\|^2+\frac1{\lambda_1}\Xi{}:\Lambda.$$
is clearly convex with respect to $(H,\Lambda)$, takes nonnegative values, and is such that $\widehat P_{\mathrm{diss}}(F,0,\Xi{},0)=0$. This function is a dissipation potential for $\widehat T_{\mathrm{Rd}}$\footnote{The dissipation potential for the compressible Newtonian fluid in Lagrangian form we wrote in Remark 3.9 of \cite{HLDAR1} is incorrect and should read $\widehat P_{\mathrm{diss}}(F,H)=\nu\det F\|\Sym(HF^{-1})\|^2$.} and  $\widehat K$, in a slightly generalized sense compared to \cite{HLDAR1} allowing for more flexibility in the arguments of the potential, namely that we have
 $\widehat T_{\mathrm{Rd}}(F,H,\Xi{})=\frac{\partial \widehat P_{\mathrm{diss}}}{\partial H}\bigl(F,H,\Xi{},\frac{\partial\widehat A_m}{\partial \Xi{}}(F,\Xi{})\bigr)$
and
  $\widehat K(F,H,\Xi{})=-\frac{\partial  \widehat P_{\mathrm{diss}}}{\partial\Lambda}\bigl(F,H,\Xi{},\frac{\partial\widehat A_m}{\partial \Xi{}}(F,\Xi{})\bigr)$, thus yielding a nonnegative dissipation. Of course, this remark adds very little insight into this particular situation compared to our direct approach.
\end{remark}

\section{A few variants of the  Oldroyd B model}\label{pour quelques modeles de plus}

In this section, we develop a few generalizations of the Oldroyd B model based on Lagrangian formulations and discuss their relationship with the second law. 

\subsection{Nonlinear Oldroyd B models}\label{nonlinear oldB}

We now define a whole family of nonlinear incompressible Oldroyd B-type models in the Lagrangian description. We keep $\widehat A_m$ based on the neo-Hookean material as in \eqref{KVM ansatz}, so that $\sigma_p=\mu F\Xi{}F^T$ as in \eqref{polymer stress}, but consider more elaborate flow rules. Specifically, given $k\in\N$, we set 
\begin{equation}\label{flow rule nl}
\widehat K_k(F,H,\Xi{})=-\frac{\mu^k}{\lambda_1\mu_k}\Xi{}(C\Xi{})^{k}+\frac{2\eta_p}{\mu\lambda_1}F^{-1}\Sym(HF^{-1})F^{-T}.
\end{equation}
For $k=0$ and $\mu_0=1$, this is the linear flow rule \eqref{la bonne flow rule} which gives rise to the usual Oldroyd B model. The constant $\mu_k>0$ is a physical parameter that is homogeneous to a pressure to the power $k$. Indeed, $C$ and $\Xi{}$ are dimensionless whereas the shear modulus $\mu$ is homogenous to a pressure. 

\begin{proposition}\label{sigma p non lineaire}
The Eulerian polymer stress corresponding to the data \eqref{KVM ansatz}, \eqref{newtonian}, and \eqref{flow rule nl} satisfies
\begin{equation}\label{stress p nl}
\overset{\triangledown}{\sigma}_p=-\frac1{\lambda_1\mu_k}\sigma_p^{k+1}+\frac{2\eta_p}{\lambda_1}d.
\end{equation}
\end{proposition}
\begin{proof}
We proceed exactly as in the proof of Proposition \ref{jackpot!} by substituting the flow rule \eqref{flow rule nl} into equation \eqref{derivee Old B du stress p}. The second term in the flow rule gives rise to the second term in the right-hand side of \eqref{stress p nl} as before.  The first term becomes
$$
\mu F\Bigl(-\frac{\mu^k}{\lambda_1\mu_k}\Xi{}(C\Xi{})^{k}\Bigr)F^T=-\frac1{\lambda_1\mu_k}\bigl(\mu F\Xi{}F^T\bigr)^{k+1}
=-\frac1{\lambda_1\mu_k}\sigma_p^{k+1}
$$
by equation \eqref{polymer stress}.
\end{proof}

We thus obtain nonlinear incompressible Oldroyd B models by writing $\sigma=\sigma_s+\sigma_p$ with $\sigma_s=2\eta_s d$.
\begin{corollary}\label{modele non lineaire}
The Eulerian form of the above Lagrangian model reads
\begin{equation}\label{equa diff Old B nl}
\frac1{\mu_k}(\sigma-2\eta_s d)^{k+1}+\lambda_1\overset{\triangledown}{\sigma}=2\bigl(\eta_pd+\lambda_1\eta_s\overset{\triangledown}{d}\bigr).
\end{equation}
\end{corollary}

\begin{remark}These models are frame-indifferent by construction. This also follows from the fact that equation \eqref{equa diff Old B nl} can be rearranged as
\begin{equation*}\label{nl general}
\lambda_1\overset{\pentagon_1}{\sigma}=2\eta\lambda_2\overset{\pentagon_2}{d},
\end{equation*}
where $\eta=\eta_s+\eta_p$ and $\lambda_2=\lambda_1\eta_s/\eta$ as before, with more or less complicated objective derivatives $\bigpenta_i$.

 There is no nice general expression for $\bigpenta_i$ because of the noncommutative binomial expression. It however yields a symmetric-valued polynomial in $(\sigma,d)$, which is an objective function, see Section \ref{les derivees objectives}. Let us expand the simplest nonlinear example.
 
 For $k=1$, since $
(\sigma-2\eta_s d)^{2}= \sigma^2-2\eta_s(\sigma d+d\sigma)+4\eta_s^2 d^2
$,
it follows that \eqref{equa diff Old B nl} can be rearranged as
$$
\sigma^2-2\eta_s(\sigma d+d\sigma)+\mu_1\lambda_1\overset{\triangledown}{\sigma}=
2\Bigl(\mu_1\eta_pd-2\eta_s^2d^2+\mu_1\lambda_1\eta_s\overset{\triangledown}{d}\Bigr).
$$
%
 When $\eta_p=0$, \emph{i.e.}, $\lambda_1=\lambda_2$, and $\mu_1=1$, this is precisely the quadratic model obtained by \cite{GFOLP}.


All these models involve nonlinear differential equations for $\Xi{}$ or $\sigma_p$ as soon as $k\ge 1$, whereas the corresponding differential  equations are linear for $k=0$. Consequently, given a prescribed, smooth enough deformation $\phi(X,t)$ and an initial value for $\Xi{}$ or $\sigma_p$, the existence of the internal variables is  a priori only ensured locally in time when $k\ge 1$. It is not clear that they exist globally in time. 
\begin{remark}
In the quadratic case $k=1$, we see that $Z=\frac\mu{\lambda_1\mu_1}\Xi{}$ is a solution of the matrix Riccati equation $Z'+ZCZ=G$, with $G=-\frac{\eta_p}{\lambda_1^2\mu_1}(C^{-1})'$. If $G=0$, that is to say $\eta_p=0$ as in \cite{GFOLP}, or $C$ constant in time, then the homogeneous Riccati equation is classically solved as $Z(t)=Z(0)\bigl(I+\bigl(\int_0^tC(s)\,ds\bigr)Z(0)\bigr)^{-1}$, provided the matrix between parentheses is invertible. This is the case when $Z(0)$ is positive semi-definite. Indeed, let $Y(t)=I+\bigl(\int_0^tC(s)\,ds\bigr)Z(0)$. Consider $u\in \ker Y(t)$. We thus have
$$0=u+\Bigl(\int_0^tC(s)\,ds\Bigr)Z(0)u.$$ Multiplying this on the left by $u^TZ(0)$, we obtain $$0=u^TZ(0)u+u^TZ(0)\Bigl(\int_0^tC(s)\,ds\Bigr)Z(0)u.$$ Both terms are nonnegative, hence 
$u^TZ(0)\bigl(\int_0^tC(s)\,ds\bigr)Z(0)u=0$.
 But $C(s)$ is positive definite for all $s$, and so is its integral. It follows that $Z(0)u=0$, and thus $u=0$.
 
 Conversely, if we assume that $Z(0)$ has at least one strictly negative eigenvalue, taking $C(t)=I$ makes $Y(t)$ become non invertible in finite time, which corresponds to blowup of the Riccati equation solution. 
So positive semi-definiteness of the initial condition is a necessary and sufficient condition for global existence of the internal variable in these particular cases.

If on the other hand $G\neq 0$, it is possible to give examples of $C(t)$ such that the associated Riccati equation with positive initial conditions blows up in finite time. Sufficient conditions for global existence may possibly be obtained by optimal control arguments.
\end{remark}
More generally, by taking linear combinations of the first terms of several flow rules of the form \eqref{flow rule nl}, we obtain models of the form
\begin{equation*}\label{encore plus general}
\overset{\triangledown}{\sigma}_p=-\frac1{\lambda_1}\sigma_pP(\sigma_p)+\frac{2\eta_p}{\lambda_1}d,
\end{equation*}
where $P\in\R[X]$ is any polynomial. For instance, $P(X)=1+\varepsilon X$, $\varepsilon>0$ small, would correspond to a quadratic correction of the classical linear Oldroyd B model.
\end{remark}

Let us now discuss second law issues for these nonlinear models, which is again made possible by the Lagrangian point of view. We just consider the monomial models above.

\begin{proposition}The constitutive law for the internal dissipation corresponding to the flow rule $\widehat K_k$ is given by 
$$
 \widehat D_{\mathrm{int}}(F,H,\Xi{})=2\eta_s\|\Sym(HF^{-1})\|^2+\frac{\mu^{k+1}}{2\lambda_1\mu_k}\tr\bigl((C\Xi{})^{k+1}\bigr)
$$
in Lagrangian form and
 $$
 \widehat d_{\mathrm{int}}(h,\sigma_p)=2\eta_s\|d\|^2+\frac1{2\lambda_1\mu_k}\tr\bigl( \sigma_p^{k+1}\bigr)
 $$
 in Eulerian form.
\end{proposition}
\begin{proof}
Many terms are the same as before, we only focus on the one that is different, namely
$$
\frac\mu2 C:\Bigl(\frac{\mu^k}{\lambda_1\mu_k}\Xi{}(C\Xi{})^{k}\Bigr)=\frac{\mu^{k+1}}{2\lambda_1\mu_k}\tr\bigl((C\Xi{})^{k+1}\bigr).
$$
To translate this into Eulerian terms, we notice that 
$$
\tr\bigl((C\Xi{})^{k+1}\bigr)=\tr\bigl((F\Xi{}F^T)^{k+1}\bigr)
$$
hence the result since $\sigma_p=\mu F\Xi{}F^T$.
\end{proof}

\begin{proposition}\label{ce qui compte}
For $k$ odd, the model satisfies the second law unconditionally, \emph{i.e.}, in the sense of Coleman-Noll, as long as the internal variables exist.\end{proposition}
\begin{proof}
Indeed, $k+1$ is then even, therefore  $\tr (\sigma_p^{k+1})\ge 0$ for any $\sigma_p\in\Sym_3$.
\end{proof}
The quadratic case $k=1$ is thus unconditional. We have just noticed that positive semi-definiteness of the initial condition is nonetheless necessary and sufficient for the global existence of the internal variable in the homogeneous case. 
 
 \begin{remark}
 For $k\ge 2$ even, we do not have an explicit formula for $\Xi{}$ to work with as in the case $k=0$. For the conditional second law to hold, the initial dissipation must still be positive for all $C(0)=C\in \Sym_3^+\cap\SL(3)$. Since $\Xi{}(0)$ is symmetric, it is orthogonally diagonalizable with $\Xi{}(0)=Q^T\Delta Q$, $\Delta=\diag(v_j)$. 
  \begin{multline*}
\tr\bigl((C\Xi{}(0))^{k+1}\bigr)=\tr\bigl((CQ^T\Delta Q)^{k+1}\bigr)\\
=\tr\bigl(Q(CQ^T\Delta Q)^{k+1}Q^T\bigr)
=\tr\bigl((QCQ^T\Delta)^{k+1}\bigr).
\end{multline*}
Now, as in the proof of Proposition \ref{condition C}, $Q^TCQ$ can in fact be any matrix in $\Sym_3^+$, not just those in $\SL(3)$, in particular $C=\diag(c_j)$, $c_j>0$. It follows that we must have
$\sum_{j=1}^3c_j^{k+1}v_j^{k+1}\ge 0$, hence $v_j^{k+1}\ge 0$ for all $j$. Since $k+1$ is odd, it is necessary that $v_j\ge 0$, so that $\Xi{}(0)$ must be positive semi-definite.

 Conversely, for any $C\in\Sym_3^+$, 
 $\tr\bigl((C\Xi{}(0))^{k+1}\bigr)=\bigl((C\Xi{}(0))^{k}C\bigr):\Xi{}(0)$. Since $k$ is even, it is clear that for all $B\in \Sym_3$,  $(CB)^{k}C\in\overline{\Sym_3^+}$. Therefore, if $\Xi{}(0)$ is positive semi-definite, then 
 $\tr\bigl((C\Xi{}(0))^{k+1}\bigr)\ge 0$. 
 
 If $\Xi{}(0)$ is positive definite, then it will remain so at least for some time, and the dissipation is initially strictly positive. It is not clear that the dissipation stays nonnegative for as long as $\Xi{}$ exists.
 \end{remark}
 
\subsection{The Zaremba-Jaumann and Oldroyd A fluids}\label{ZJ and co}

We finally consider complex fluid models based on two other objective derivatives, the Zaremba-Jaumann fluid, see for instance \cite{Eiter}-\cite{LionsMasmoudi}, and the Oldroyd A fluid, \cite{Hinch}-\cite{Oldroyd}-\cite{Renardy}, both models being considerably less prominent in the literature than  the Oldroyd B fluid.

 Expressed in terms of the polymer stress, they simply read
\begin{equation*}\label{ZB fluid}
\sigma_p+\lambda_1\overset{\square}{\sigma}_p=2\eta_p d,
\end{equation*}
for Zaremba-Jaumann and 
\begin{equation*}\label{Old A fluid}
\sigma_p+\lambda_1\overset{\vartriangle}{\sigma}_p=2\eta_p d,
\end{equation*}
for Oldroyd A,
together with a Newtonian solvent stress.

We can derive both models from our Lagrangian formulation  by adapting the flow rule while retaining \eqref{KVM ansatz}, \eqref{newtonian} and the incompressibility condition. It is to be expected that they are slightly less natural than the Oldroyd B model, due to Proposition \ref{et voila OB}.

We thus use the same ingredients as for the Oldroyd B model, except  for the flow rule, with
\begin{multline*}
\widehat K_{\textnormal{ZJ}}(F,H,\Xi{})=-\frac1{\lambda_1}\Xi{}+\frac{2\eta_p}{\mu\lambda_1}F^{-1}\Sym(HF^{-1})F^{-T}\\
-F^{-1}\Sym(HF^{-1})F\Xi{}-\Xi{}F^T\Sym(HF^{-1})F^{-T},
\end{multline*}
for Zaremba-Jaumann, and 
\begin{multline*}
\widehat K_{\textnormal{A}}(F,H,\Xi{})=-\frac1{\lambda_1}\Xi{}+\frac{2\eta_p}{\mu\lambda_1}F^{-1}\Sym(HF^{-1})F^{-T}\\
-2F^{-1}\Sym(HF^{-1})F\Xi{}-2\Xi{}F^T\Sym(HF^{-1})F^{-T}.
\end{multline*}
for Oldroyd A. Note that the resulting ordinary differential equations for $\Xi{}$ are still linear, but with variable coefficients. Therefore, there is no explicit Duhamel formula expressing their solutions, as opposed to the Oldroyd B case. This also explains why we resort below to numerical simulations in order to investigate the properties of these models with respect to the second law. 

It is a simple computation to check that

\begin{proposition}The Lagrangian models produced by the above choices are the Zaremba-Jaumann fluid and Oldroyd A fluid models respectively. 
\end{proposition}

Note that the correction applied to Oldroyd B in order to obtain Oldroyd A is twice that applied to obtain Zaremba-Jaumann. Indeed, $\overset{\square}{\sigma}_p=\frac12\bigl(\overset{\triangledown}{\sigma}_p+\overset{\vartriangle}{\sigma}_p\bigr)$, see Section \ref{les derivees objectives}.
We can also compute their internal dissipations. For the Zaremba-Jaumann fluid, we obtain

\begin{proposition}The internal dissipation of the Zaremba-Jaumann model is given by 
\begin{equation}\label{dissipation ZJ}
  \widehat D_{\mathrm{int}}(F,H,\Xi{})=2\eta_s\|\Sym(HF^{-1})\|^2+\frac\mu2 \Xi{}:\Bigl(\frac{1}{\lambda_1}C+C'\Bigr)
\end{equation}
 in the Lagrangian formulation, using $C'$ as shorthand for $2F^T\Sym(HF^{-1})F$, and 
\begin{equation}\label{dissipation ZJ Euler}
  \widehat d_{\mathrm{int}}(h,\xi)=2\eta_s\|d\|^2+\frac1{2\lambda_1}\tr\xi+\xi:d
\end{equation}
in the Eulerian formulation, with $\xi=\sigma_p$ as before.
\end{proposition}

\begin{proof}
Let us just compute the part $ \widehat D_{\mathrm{int,f}}$ of the dissipation stemming from the flow rule. We still have $\frac{\partial \widehat A_m}{\partial \Xi{}}=\frac\mu2 C$. Therefore
\begin{align*}
\widehat D_{\mathrm{int,f}}(F,H,\Xi{})&=\frac\mu2 C:\Bigl(\frac1{\lambda_1}\Xi{}-\frac{2\eta_p}{\mu\lambda_1}F^{-1}\Sym(HF^{-1})F^{-T}\\&\qquad\qquad
+F^{-1}\Sym(HF^{-1})F\Xi{}+\Xi{}F^T\Sym(HF^{-1})F^{-T}\Bigr)\\
&=\frac\mu{2} \Bigl(C:\frac1{\lambda_1}\Xi{}
+2F^T\Sym(HF^{-1})F:\Xi{}\Bigr)
\end{align*}
since $\tr\bigl(\Sym(HF^{-1})\bigr)=\tr(d)=0$ as before. The translation in Eulerian terms is straightforward since $\xi=\sigma_p= \mu F\Xi{}F^T$ still.
\end{proof}

This is a case when the identification of the set $\mathcal{C}_0$ is again possible.
\begin{proposition}There holds $\mathcal{C}_0=\{0\}$.
\end{proposition}

\begin{proof}First of all, by \eqref{def de C0 bis}, \eqref{KVM ansatz}
 and \eqref{dissipation ZJ}, it is clear that $0\in\mathcal{C}_0$.

Conversely, let $\Xi_0\in\mathcal{C}_0$.
For $F\in \SL(3)$ and $k\in \R$, we take $H\in T_F\SL(3)$ such that $\Sym(HF^{-1})=k\mu D$ with
\begin{equation}\label{choix de D}
D=\Bigl(F\Xi_0 F^T-\frac13\tr(F\Xi_0 F^T)I\Bigr).
\end{equation}
By \eqref{def de C0 bis}, \eqref{KVM ansatz}
 and \eqref{dissipation ZJ} again, this yields
$$
2\eta_s\|D\|^2k^2+\|D\|^2k+\frac1{2\mu\lambda_1}\tr(F\Xi_0 F^T)\ge 0,
$$
for all $k\in\R$. The left-hand side is a second degree polynomial in $k$, hence its discriminant must be nonpositive, \emph{i.e.},
$$
\|D\|^4-\frac{4\eta_s}{\mu\lambda_1}\|D\|^2\tr(F\Xi_0 F^T)\le 0,\text{ or }
\|D\|^2\le \beta\tr(F\Xi_0 F^T),
$$
where $\beta=\frac{4\eta_s}{\mu\lambda_1}$. Now, because of \eqref{choix de D}, it follows that 
$$\|D\|^2=\|F\Xi_0 F^T\|^2-\frac13\bigl(\tr(F\Xi_0 F^T)\bigr)^2.$$
We have thus so far obtained the necessary condition for $\Xi_0$ to belong to $\mathcal{C}_0$:
\begin{equation}\label{une condition pour ZJ}
\|F\Xi_0 F^T\|^2-\frac13\bigl(\tr(F\Xi_0 F^T)\bigr)^2\le \beta\tr(F\Xi_0 F^T),
\end{equation}
for all $F\in \SL(3)$. 

Now let $Q\in\SO(3)$ be such that $Q\Xi_0Q^T=\Delta$ with $\Delta=\diag(\delta_i)$, and $U=\diag(u^{1/2},u^{-1/2},1)$ with $u>0$. We take $F=UQ$. With this choice, it follows readily that $F\Xi_0 F^T=\diag(u\delta_1,\delta_2/u,\delta_3)$ and inequality \eqref{une condition pour ZJ} becomes
$$
\delta_1^2+u^{-4}\delta_2^2+u^{-2}\delta_3^3-\frac13\bigl(\delta_1+u^{-2}\delta_2+u^{-1}\delta_3\bigr)^2
\le \beta (u^{-1}\delta_1+u^{-3}\delta_2+u^{-2}\delta_3).$$
Letting $u\to+\infty$, we obtain that $\frac23\delta_1^2\le 0$, so that $\delta_1=0$. Of course then $\delta_2=\delta_3=0$ as well and $\Xi_0=0$.
\end{proof}

It follows from the general theory that $\mathcal{C}\subset\mathcal{C}_0=\{0\}$. However, the same numerical experiment as in \cite{HLDAR1} shows that $\Xi_0=0$ can lead to strictly negative dissipation, see Figure \ref{fig} below. We therefore can say that $\mathcal{C}=\emptyset$ or that the Zaremba-Jaumann fluid with the neo-Hookean inspired free energy \eqref{KVM ansatz} does not satisfy the second law even conditionally.

Also in \cite{HLDAR1}, we showed that the alternate choice $\widehat a_m(\xi)=\frac{\lambda_1}{4\eta_p}\|\xi\|^2$ as free energy, makes the Zaremba-Jaumann fluid  satisfy the second law unconditionally, using the corresponding dissipation. Unfortunately, as already noticed in \cite{HLDAR1}, this choice is not physical due to the necessity of the factor $\frac{\lambda_1}{4\eta_p}$ in an elastic energy term, a factor which is built on constants from the flow rule pertaining to viscous effects. A translation in Lagangian terms would therefore satisfy the second law unconditionally, but suffer from the same defect.

The situation is the same for the Oldroyd A fluid, based on the remark above on the corrections applied to Oldroyd B. 

\begin{proposition}The internal dissipation of the Oldroyd A model is given by 
\begin{equation}\label{dissipation OA}
  \widehat D_{\mathrm{int}}(F,H,\Xi{})=2\eta_s\|\Sym(HF^{-1})\|^2+\frac\mu2 \Xi{}:\Bigl(\frac{1}{\lambda_1}C+2C'\Bigr)
\end{equation}
in the Lagrangian formulation and 
\begin{equation}\label{dissipation OA Euler}
  \widehat d_{\mathrm{int}}(h,\xi)=2\eta_s\|d\|^2+\frac1{2\lambda_1}\tr\xi+2\xi:d
\end{equation}
in the Eulerian formulation, still with $\xi=\sigma_p$.

In this case, there also holds $\mathcal{C}_0=\{0\}$.
\end{proposition}

The same negative considerations concerning the second law hold for the Oldroyd A model. This does not rule out other physically motivated choices of free energies for each of the two models, for which the second law could be conditionally satisfied.

In Figure \ref{fig} below, we show the results of a numerical simulation for both Oldroyd models and for the Zaremba-Jaumann model, with the same data, already described in \cite{HLDAR1}. Namely, an Eulerian computation with $h(x,t)=\cos(\omega t)m$, where $m$ is a randomly chosen $3\times 3$ traceless matrix with coefficients between $-1$ and $1$, and $\omega=0.75$. The material constants are $\lambda_1=10$, $\eta_s=.1$ and $\eta_p=1.9$. The initial value for $\sigma_p=\xi$ is $\xi(0)=0$, in order to accommodate all three $\mathcal{C}_0$ sets. We plot the internal dissipation $d_{\mathrm{int}}(t)$ with a solid line, $\frac1{2\lambda_1}\tr\xi(t)$ with a dashed line (Oldroyd A and B cases only\footnote{It is easy to see that $\tr\xi(t)=0$ in the Zaremba-Jaumann case}) and $\xi(t):d(t)$ with a dotted line (Zaremba-Jaumann and Oldroyd A cases only),  for $t$ from $0$ to $40$, with different vertical scales for each model. We see that the Oldroyd B dissipation \eqref{dissipation eulerienne} remains nonnegative as expected, whereas both Zaremba-Jaumann dissipation \eqref{dissipation ZJ}-\eqref{dissipation ZJ Euler} and Oldroyd A dissipation \eqref{dissipation OA}-\eqref{dissipation OA Euler} take strictly negative values in finite time, thus violating the conditional second law. 

\begin{figure}[ht]
\begin{center}
  \includegraphics[scale=.2242]{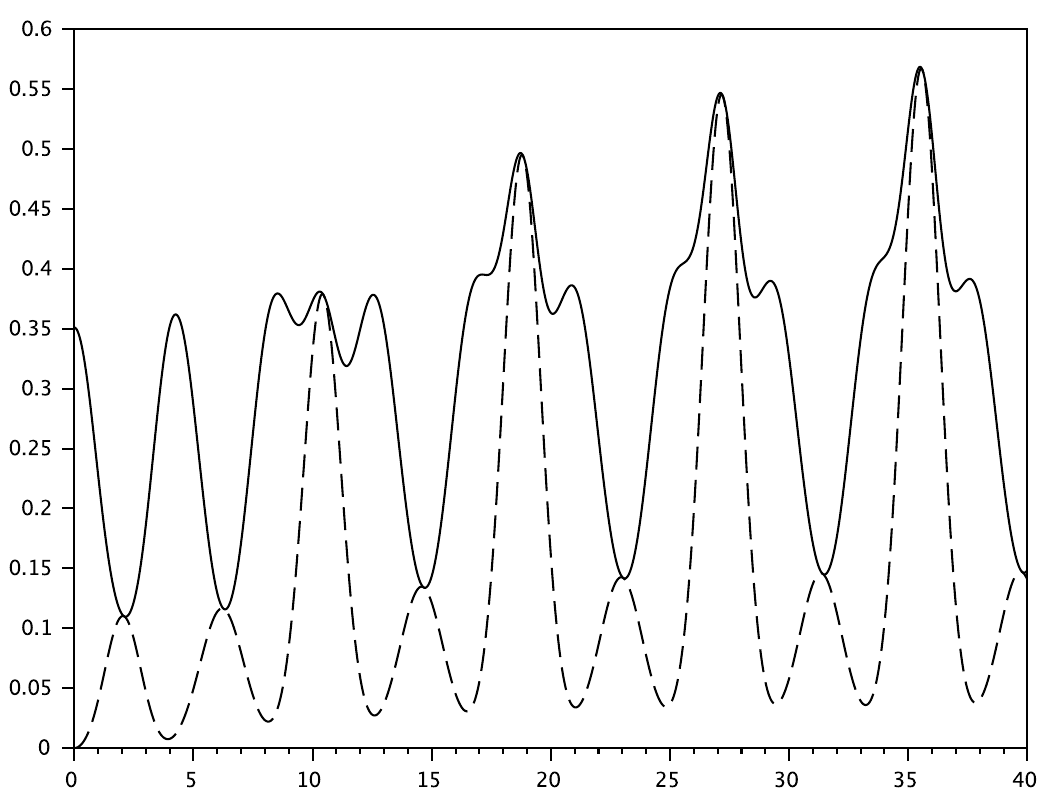}
  \includegraphics[scale=.2242]{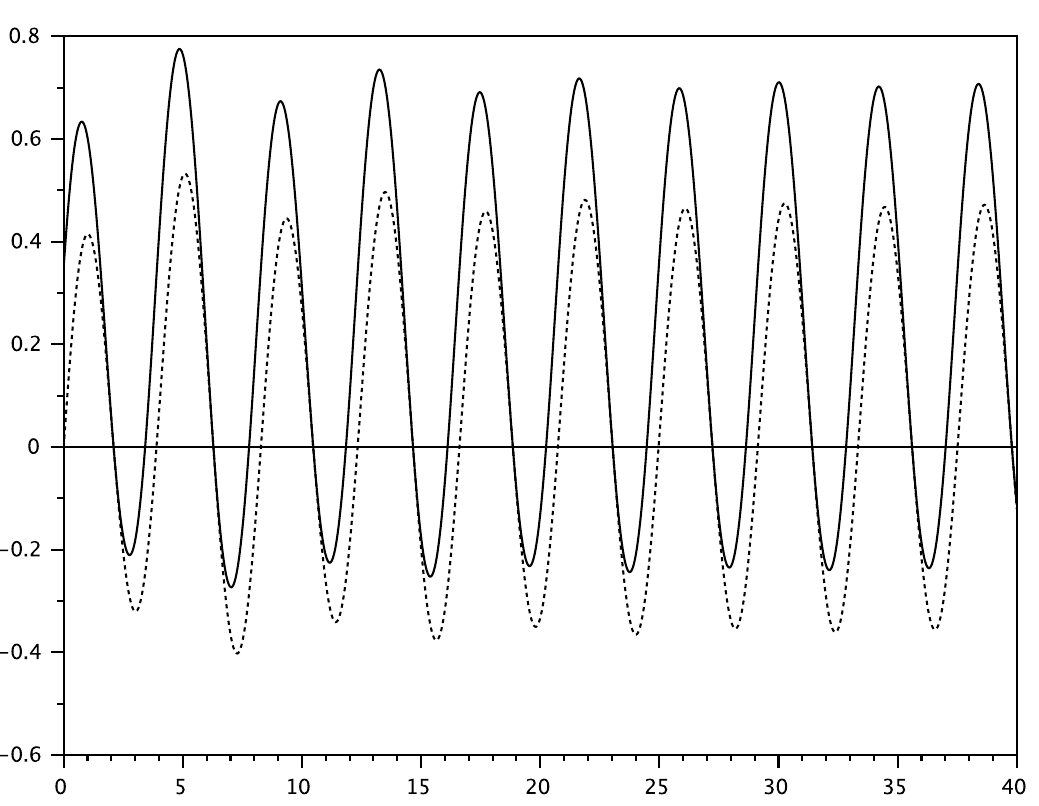}
   \includegraphics[scale=.2242]{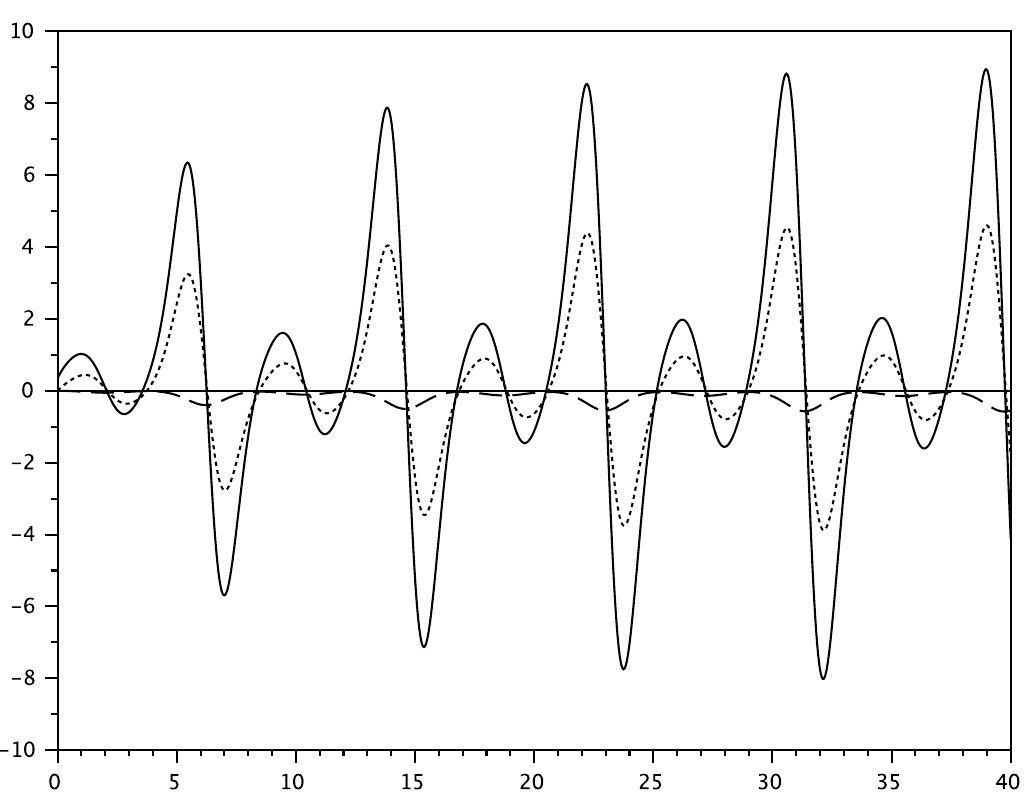}
\caption{Left: Oldroyd B, center: Zaremba-Jaumann, right: Oldroyd A.\label{fig}}
\end{center}
\end{figure}

In Figure \ref{fig2}, we plot the smallest eigenvalue of $\sigma_p$ vs.\ time in the above Oldroyd B case. We see that this smallest eigenvalue becomes strictly negative, hence neither $\sigma_p$ nor $\Xi{}$ remain positive semi-definite for all times in this particular example either, \emph{cf.}\ Proposition \ref{damned}. Note that here $-\frac{\eta_p}{\lambda_1}=-.19$, and the smallest eigenvalue of $\sigma_p$ stays above this value as expected from Remark \ref{borne inf}.

\begin{figure}[ht]
\begin{center}
  \includegraphics[scale=.3]{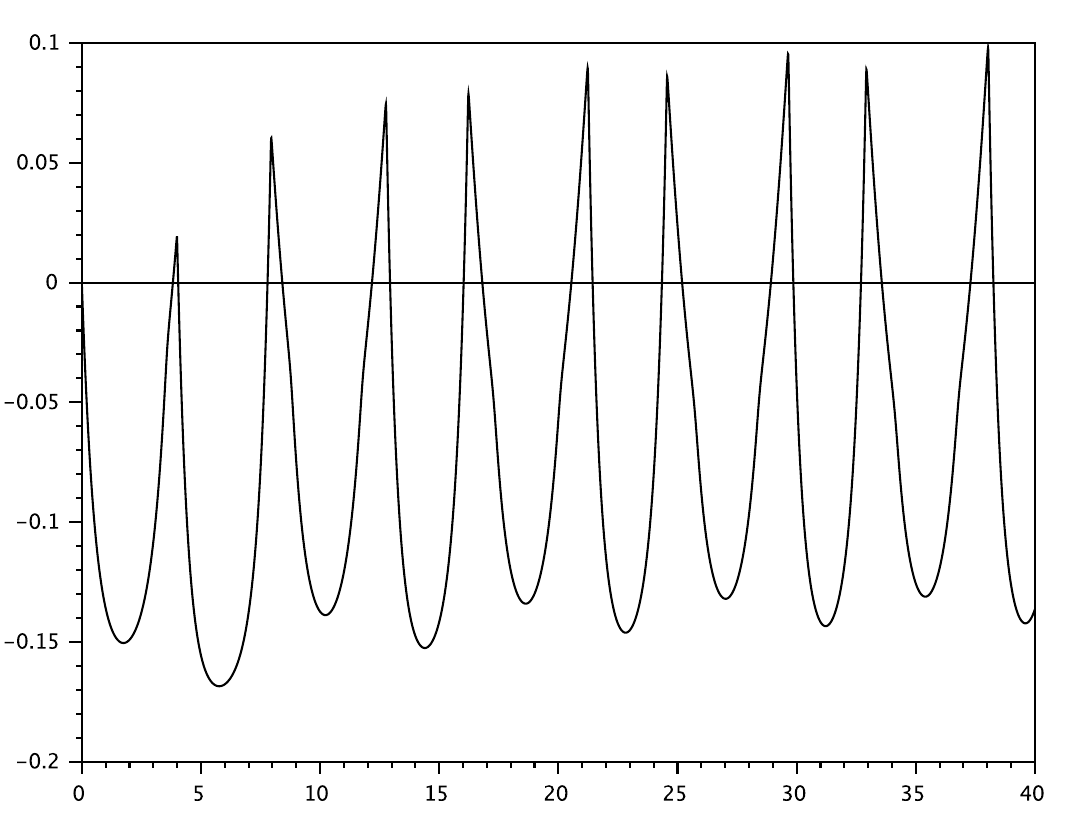}
\caption{Smallest eigenvalue of $\sigma_p$, Oldroyd B case.\label{fig2}}
\end{center}
\end{figure}

\end{document}